\documentclass[final,1p,times]{elsarticle}

\usepackage{amssymb}
\usepackage{amsmath}
\usepackage{amsthm}

\newtheorem{prop}{Proposition}
\newtheorem{coro}{Corollary}
\newtheorem{definition}{Definition}%

\usepackage{graphicx}%
\usepackage{multirow}%
\usepackage{amsmath,amssymb,amsfonts}%
\usepackage{mathrsfs}%
\usepackage[title]{appendix}%
\usepackage{xcolor}%
\usepackage{textcomp}%
\usepackage{manyfoot}%
\usepackage{booktabs}%
\usepackage{algorithm}%
\usepackage{algorithmicx}%
\usepackage{algpseudocode}%
\usepackage{listings}%
\usepackage{mathtools}
\usepackage[english]{babel}
\usepackage{enumitem}
\usepackage{placeins}
\usepackage{subcaption}

\newcommand{\manif}{\mathcal{M}}
\newcommand{\xbf}{\mathbf{x}}
\newcommand{\Abf}{\mathbf{P}}

\newcommand{\ybf}{\mathbf{y}}

\newcommand{\sbf}{\mathbf{s}}
\newcommand{\cbf}{\mathbf{c}}
\newcommand{\Mbf}{\mathbf{M}}
\newcommand{\Fbf}{\mathbf{F}}
\newcommand{\Sbf}{\mathbf{S}}
\newcommand{\phibf}{\boldsymbol{\phi}_0}
\newcommand{\phitest}{\boldsymbol{\phi}_0^\text{test}}

\newcommand{\phibfn}{\phibf^{\text{obs}}}
\newcommand{\sigbf}{\boldsymbol{\Sigma}}
\newcommand{\zbf}{\mathbf{z}}
\newcommand{\Zbf}{\mathbf{Z}}
\newcommand{\qbf}{\mathbf{Q}}
\newcommand{\vbf}{\mathbf{v}}
\newcommand{\Vbf}{\mathbf{V}}
\newcommand{\Kbf}{\mathbf{K}}

\newcommand{\ombf}{\boldsymbol{\omega}}
\newcommand{\Ibf}{\mathbf{I}}

\newcommand{\ubf}{\mathbf{u}}
\newcommand{\Lbf}{\mathbf{L}}

\newcommand{\param}{\boldsymbol{\beta}}

\newcommand{\lm}{L^2(\manif)}
\newcommand{\hm}{H^2(\manif)}
\newcommand{\mug}{\mu_g}
\newcommand{\eqdef}{\vcentcolon=}

\newcommand{\kbf}{\mathbf{k}}
\newcommand{\tri}{\mathrm{T}}

\newcommand{\diag}{\text{Diag}}
\newcommand{\alp}{\alpha}
\newcommand{\Scal}{\mathcal{S}}
\newcommand{\Gbf}{\mathbf{G}}
\newcommand{\Rbf}{\mathbf{R}}
\newcommand{\hh}{\mathbf{h}}
\newcommand{\Dbf}{\mathbf{D}}

\newcommand{\llik}{\mathcal{L}}
\newcommand{\llikc}{\mathcal{L}^{(C)}}

\newcommand{\amnoise}{a_{\noise}^{(m)}}
\newcommand{\ambis}{\tilde{a}_{\noise}^{(m)}}

\newcommand{\noise}{\tau}

\newcommand{\Ebf}{\mathbf{E}}

\newcommand{\anoise}{a_\noise}\newcommand{\wbf}{\mathbf{w}}

\newcommand{\Pcal}{\mathcal{P}}
\newcommand{\esp}{\mathbb{E}}
\newcommand{\var}{\mathbb{V}}

\newcommand{\mbf}{\mathbf{m}}
\newcommand{\Ktau}{\Kbf_{1,\tau}}
\newcommand{\Kzero}{\Kbf_{1,0}}

\newcommand{\cov}{\mathrm{Cov}}
\newcommand{\Yobs}{\mathbf{Y}_{\text{obs}}}
\newcommand{\Kbis}{\Kbf^{(m)}_\noise}
\newcommand{\Ktest}{\Kbf^{\text{test}}_1}
\newcommand{\vuk}{V_{\text{UK}}}
\newcommand{\mpost}{\mbf_{\text{post}}}
\newcommand{\mpostalp}{\mbf_{\text{post}}^{(\alp)}}
\newcommand{\covpost}{\sigbf_{\text{post}}}
\newcommand{\Zres}{\Zbf_\text{res}}
\newcommand{\Zpost}{\Zbf_\text{post}}
\newcommand{\zpost}{\zbf_\text{post}}
\newcommand{\Zprior}{\Zbf_a}

\newcommand{\sigalp}{\tilde{\sigbf}_\alp}
\newcommand{\qalp}{\tilde{\qbf}_\alp}
\newcommand{\gaussbfm}{\boldsymbol{\varepsilon}_m}
\newcommand{\gaussbfn}{\boldsymbol{\varepsilon}_n}
\newcommand{\veczerm}{\mathbf{0}_m}
\newcommand{\vecunm}{\mathbf{1}_m}
\newcommand{\bbf}{\mathbf{b}}
\newcommand{\Ztest}{\Zbf_\text{test}}
\newcommand{\moy}{m_\alpha}
\newcommand{\covk}{k_\alpha}
\newcommand{\Zsk}{\Zbf_{\text{SK}}}
\newcommand{\zsk}{\zbf_{\text{SK}}}
\newcommand{\msk}{\mbf_{\text{SK}}}
\newcommand{\covsk}{\sigbf_{\text{SK}}}
\newcommand{\yhat}{\hat{y}}
\newcommand{\ytrue}{y_\text{true}}
\newcommand{\sigmahat}{\hat{\sigma}}
\newcommand{\Kalp}{\Kbf_{1,\alpha}}
\newcommand{\mtest}{\mbf_\text{lim}}
\newcommand{\ktest}{\Kbf_\text{lim}}

\journal{Nuclear Physics B}

\begin{document}

\begin{frontmatter}



\title{Uncertainty Quantification of Spline Predictors on Compact Riemannian Manifolds} 


\author{Charlie Sire}\corref{cor1}
\ead{charlie.sire@minesparis.psl.eu}
\author{Mike Pereira}

\cortext[cor1]{Corresponding author}
\affiliation{organization={Center for Statistics and Images, Mines Paris, PSL University},
            city={Fontainebleau},
            postcode={77300}, 
            country={France}}

\begin{abstract}
To predict smooth physical phenomena from observations, spline interpolation provides an interpretable framework by minimizing an energy functional associated with the Laplacian operator. This work proposes a methodology to construct a spline predictor on a compact Riemannian manifold, while quantifying the uncertainty inherent in the classical deterministic solution. Our approach leverages the equivalence between spline interpolation and universal kriging with a specific covariance kernel. By adopting a Gaussian random field framework, we generate stochastic simulations that reflect prediction uncertainty. However, on compact manifolds, the covariance kernel depends on the generally unknown spectrum of the Laplace-Beltrami operator. To address this, we introduce a finite element approximation based on a triangulation of the manifold. This leads to the use of intrinsic Gaussian Markov Random Fields (GMRF) and allows for the incorporation of anisotropies through local modifications of the Riemannian metric. The method is validated using a temperature study on a sphere, where the operator's spectrum is known, and is further extended to a test case on a cylindrical surface.
\end{abstract}

\end{frontmatter}



\newpage
	\section{Introduction}
\label{intro}

Thin-plate spline interpolation \cite{wahba1990spline} is one of the most widely used methods for predicting spatially distributed data. It consists in constructing an interpolating function that exactly matches the observed values at the measurement locations while being as smooth as possible according to a precise mathematical criterion. This notion of optimality, defined through the minimization of a roughness functional, is particularly well suited to settings with few observations in which the underlying phenomenon is expected to vary smoothly over space. However, thin-plate splines yield a single deterministic prediction over the domain and do not directly provide a measure of uncertainty. In many applications, uncertainty quantification is crucial, for instance to assess risk or compute predictive probabilities, which motivates the need for a full predictive distribution rather than a point estimate \cite{santner2003design}. From this perspective, the spline predictor can be shown to be equivalent to the universal kriging predictor, as first established by Matheron \cite{matheron} and later exploited for data defined on $\mathbb{R}$ \cite{DUBRULE1984327} and on $\mathbb{R}^d$ \cite{myers1988interpolation, myers1992kriging}. 

In this framework, the observations are viewed as a realization of a random field with an unknown mean function and a prescribed covariance structure, that are related to the eigenvalues and eigenvectors of the Laplacian operator $-\Delta$. Universal kriging then relies on a statistical model to provide not only a prediction but also an associated predictive variance at every point in the domain. To fully characterize this predictive uncertainty, a Gaussian assumption is typically adopted for the prior random field, leading to Gaussian process regression (GPR) \cite{williams2006gaussian}, which enables computation of the posterior distribution of the random field, that is, the conditional distribution given the observed values at the measurement locations. While the predictive mean and variance obtained from kriging and GPR are identical, the latter framework additionally allows for the simulation of posterior realizations of the random field \cite{hewing2020simulation}. These simulations provide a richer representation of uncertainty and facilitate the visualization and quantification of plausible scenarios under the assumed statistical model. Thus, exploiting Gaussian process regression with the appropriate covariance kernel makes it possible to incorporate uncertainty into the spline prediction problem.

However, when the domain of interest is a non-planar surface, several practical challenges arise. Studies on such non-Euclidean domains appear in a variety of contexts, for example in fluid mechanics when analyzing flows around cylinders \cite{fournier2005controle, gao2018numerical}, or in neuroimaging when collecting data on the brain \cite{lila2016smooth}. In theory, the equivalence between spline interpolation and universal kriging prediction still holds on compact Riemannian manifolds, as formally established in \cite{kim2001splines} using the framework of Reproducing Kernel Hilbert Spaces (RKHS), building on the pioneering work of Wahba on the sphere \cite{wahba1981spline}. On manifolds, the mean function and covariance kernel underlying this equivalence are derived from the spectral decomposition of the Laplace-Beltrami operator, a natural generalization of the Euclidean Laplacian. In practice, however, the eigenvalues and eigenfunctions of the Laplace-Beltrami operator are generally unknown for surfaces that are neither planar nor spherical. This lack of spectral information prevents the exact construction of spline predictors and their associated covariance structures for data defined on such manifolds.\\~\

To overcome this limitation, and as proposed in \cite{sire2025spline}, one possible approach is to rely on a finite element approximation of random fields defined on compact manifolds, thereby working with a finite set of approximate eigenvalues and eigenfunctions. This type of approximation was first introduced by Lindgren \cite{Lindgren} and later extended by Pereira \cite{pereira:tel-02499376} to incorporate local anisotropies into the covariance structure. This framework can then be adapted to construct spline predictors with preferred directions of dependence, which can significantly improve prediction accuracy when the underlying phenomenon is anisotropic.

However, a key difficulty arising in the finite-element approximation of spline predictors is that it requires working with a Gaussian Markov random field (GMRF, see~\cite{rue2005gaussian}) whose covariance matrix is not invertible, referred to as an intrinsic GMRF. This type of model, investigated in particular in~\cite{bolin2025intrinsic}, complicates both the formulation of the classical kriging equations and the generation of simulations. As shown in~\cite{simpson2006sampling} and thoroughly studied in~\cite{bolin2021efficient}, such simulations can be obtained through linear conditioning of non-intrinsic GMRFs. This strategy enables an extension of the work of~\cite{sire2025spline}, which focused exclusively on computing an approximate solution to the spline interpolation problem, without addressing the associated uncertainty quantification of the resulting predictions.\\~\

Therefore, the contribution of this work is to propose a statistical model for the spline interpolation problem, with the aim of providing both spline-based predictions at every point of the investigated compact Riemannian manifold and a quantification of the associated uncertainty. A strategy for generating conditional simulations of the underlying random function is also proposed, enabling the visualization of predictive scenarios arising from spline interpolation problems.\\~\

The paper is organized as follows. Section~\ref{splines_manif} introduces the theoretical solution to the spline interpolation problem on compact Riemannian manifolds and establishes its equivalence with kriging. Section~\ref{sec_fe_approx} presents the finite element approximation of this solution. Section~\ref{sec_simus} then details the strategy for generating conditional simulations of the target intrinsic GMRF, while Section~\ref{sec_likeli} addresses the introduction of anisotropy and the estimation of the unknown parameters. Finally, the application test cases and corresponding results are presented in Section~\ref{sec_appli}, and Section~\ref{sec_conclu} summarizes the contributions and outlines potential directions for future work.



\section{Spline interpolation solution on compact manifolds}\label{splines_manif}

In this manuscript, the investigated domain is a smooth compact and connected manifold $\manif$ of dimension $d$, equipped with a Riemannian metric $g$, and $-\Delta$ denotes the Laplace-Beltrami operator on $(\manif, g)$. The integral of a function $f$ over $\manif$ is denoted by $\int_\manif f d\mug$, where $d\mug$ is the canonical measure associated to $(\manif, g)$ and $\partial \manif$ denotes its boundary. The spaces $\lm$ and $\hm$ denote the set of square-integrable functions and the Sobolev space of order $2$, respectively. The objective of this article is to obtain a prediction of a phenomenon based on $n$ observations denoted by $\ybf = (y_i)_{i=1}^{n}$ at points $\Scal = \left(\sbf_i\right)_{i=1}^n$ with $\sbf_i \in \manif, 1\leq i \leq n.$ 

\subsection{Formulation of the spline prediction problems}\label{formulation}

The spline interpolation problem on $\manif$, denoted by $\Pcal_0$, consists in finding the smoothest function, in a precise mathematical sense, that interpolates the observed data. More precisely, it consists in solving the following variational problem \cite{keller2019thin,duchamp2003spline}.

\begin{definition}[Spline interpolation problem on $\manif$]
$\Pcal_0$ the spline interpolation problem on $\manif$ consists in finding a function $u \in \hm$ such that
\begin{equation}
\label{eq:P0}
\left\{
\begin{aligned}
E(u) \;&=\; \min_{v \in \hm} \int_{\manif} \lvert \Delta v \rvert^2 \, d\mug\\
u(\sbf_i) \;&=\; y_i, \qquad 1 \leq i \leq n .
\end{aligned}
\right.
\end{equation}
\end{definition}

When the observed data are corrupted by noise, seeking an exact interpolating solution is no longer appropriate. In this case, the smoothing spline problem, denoted by $\Pcal_\noise$, is more suitable, as it consists in minimizing the energy functional penalized by the prediction error at the observation points.

\begin{definition}[Smoothing spline problem on $\manif$]
$\Pcal_\noise$ the smoothing spline problem on $\manif$ consists in finding, for $\tau> 0$ a function $u \in \hm$ minimizing
\begin{equation}
\label{eq:Pnoise}
\frac{1}{n}\sum_{i=1}^{n} \left(y_i-u(\sbf_i)\right)^2 + \noise^2 \int_{\manif} \lvert \Delta v \rvert^2 \, d\mug
\end{equation}
\end{definition}

The analytical solutions of these problems are very similar, and are related to the spectrum of the operator $-\Delta.$

\subsection{Analytical solution to the spline prediction problems}\label{solution}

To derive the solutions to the problems $\Pcal_0$ and $\Pcal_1$, we rely on the Spectral Theorem~\cite{Craioveanu2001}, recalled in~\ref{spectral_theorem}, which provides the countable spectrum of $-\Delta$, denoted by $\left(\lambda_k, \phi_k\right)_{k \in \mathbb{N}}$. Here, $(\lambda_k)_{k \in \mathbb{N}} \subset \mathbb{R}_+$ are the eigenvalues and $(\phi_k)_{k \in \mathbb{N}} \subset \mathcal{C}^{\infty}(\manif)$ the associated eigenfunctions, which form an orthonormal basis of $\lm$. The eigenvalue $\lambda_0 = 0$ corresponds to constant functions~\cite{urakawa1993geometry}, with $\phi_0 = \frac{1}{\int_{\manif} d\mug}.$

As explained in~\cite{kim2001splines}, the theory of RKHS is needed to obtain the solution of problems of the type $\Pcal_\tau$ for $\tau \geq 0.$ Some details were recalled in~\cite{sire2025spline}, but are not evoked here. We simply introduce the kernel that will be central in the solution, which is defined as 

\begin{equation}
K_1(\sbf_1, \sbf_2) = \sum_{k\in \mathbb{N}^\star}\frac{1}{\lambda_k^2} \phi_k(\sbf_1)\phi_k(\sbf_2), ~~(\sbf_1,\sbf_2) \in \manif^2,
\end{equation} 

\noindent and the solution to the spline prediction problems is provided by Proposition~\ref{theo_splines}.

\begin{prop}[Spline prediction~\cite{wahbabook}]\label{theo_splines}

The solution to the problem $\Pcal_\noise$, which corresponds to the interpolating spline when $\noise = 0$ and to the smoothing spline when $\noise > 0$, is given by

\begin{equation}
u^\star_\noise(\sbf) = \anoise\phi_0(\sbf) + \kbf_1(\sbf)^\top\Ktau^{-1}\left(\ybf-\anoise\phibfn\right),
\end{equation}
with 
\begin{equation}\label{def_mat_theorem}
\left\{
\begin{aligned}
\Ktau\ &= \left[K_1(\sbf_i,\sbf_j)\right]_{1 \leq i,j \leq n} + \tau^2\Ibf_n\\
\kbf_1(\sbf) &= \left(K_1(\sbf_1,\sbf),\dots,K_1(\sbf_n,\sbf)\right)^\top ~~\forall\sbf \in \manif\\
\phibfn &= \left(\phi_0(\sbf_1),\dots,\phi_0(\sbf_n)\right)\\
\anoise &= \left((\phibfn)^\top\Ktau^{-1}\phibfn\right)^{-1}(\phibfn)^\top\Ktau^{-1} \ybf
\end{aligned}\right.
\end{equation}
\end{prop}

The invertibility of the matrix $\Ktau$ is discussed in~\ref{inversibility_ktau}. For $\tau = 0,$ $\Ktau$ is positive definite if and only if any observation vector $\ybf$ can be interpolated by a function in $\lm$, while the inversibility is guaranteed when $\tau >0.$

\subsection{Equivalence with kriging}\label{equiv_kriging}

The solution $u_\noise$ coincides with the universal kriging predictor \cite{matheron1969krigeage,hengl2009practical}, with mean $a\phibf$, where $a \in \mathbb{R}$ is an unknown parameter, covariance kernel $\sigma^2 K_1$, and observation noise variance $\sigma^2\noise^2$ when $\noise > 0$. Here, $\sigma^2 > 0$ is the unknown variance. Note that since $\phi_0$ is constant in this setting, the model reduces to ordinary kriging \cite{wackernagel2003ordinary}, which assumes a constant but unknown mean. 

Under these assumptions, and with predictive mean equal to $u_\noise(\sbf_0)$ at a new point $\sbf_0 \in \manif$, the associated kriging variance is given by

\begin{equation}
\vuk(\sbf_0) = \sigma^2\left[v(\sbf_0) + K_1(\sbf_0, \sbf_0) - k_1(\sbf_0)^\top\Ktau^{-1}k_1(\sbf_0)\right]
\end{equation} with $v(\sbf_0) = \left[\phi_0(\sbf_0)- \kbf_1(\sbf_0)^\top\Ktau^{-1}\phibfn\right]^\top\left[\left(\phibfn\right)^\top\Ktau^{-1}\phibfn\right]^{-1}\left[\phi_0(\sbf_0) - \kbf_1(\sbf_0)^\top\Ktau^{-1}\phibfn\right].$

The term $\sigma^2v(\sbf_0)$ reflects the uncertainty associated with the unknown coefficient $a$ appearing in the mean of the random field, while $\sigma^2\left[K_1(\sbf_0,\sbf_0)-k_1(\sbf_0)^\top\Ktau^{-1}k_1(\sbf_0)\right]$ corresponds to the predictive variance of the simple kriging predictor, that is, when the mean is assumed to be known.

This analogy with universal kriging naturally paves the way for uncertainty quantification in spline prediction. However, to obtain the full predictive distribution rather than just the mean and variance, the idea is to introduce Gaussian assumptions and to work with Gaussian random fields (GRFs) $Z$ and $E$. Since the parameter $a$ governing the mean is unknown, the Bayesian kriging framework~\cite{diggle2002bayesian} is adopted by modeling it as a random variable $A$. The observations $\ybf$ are then interpreted as noisy realizations of the GRF $Z$ such that,
\begin{equation*}
Y_i = Z(\sbf_i) + \sigma\tau E_i,~~1\leq i \leq n
\end{equation*}\noindent where
\begin{itemize}
\item $\left(Z \mid A=a\right)$ is a GRF with mean $a\phi_0$ and covariance kernel $\sigma^2 K_1$, for all $a\in \mathbb{R}$
\item $A \sim \mathcal{N}(\mu,\alpha),$ which is the \emph{a priori} distribution of the unknown coefficient of the mean 
\item $\Ebf = \left(E_i\right)_{i=1}^{n}$ is a centered Gaussian vector with $\cov(\Ebf) = \Ibf_n$, independent of $A$ and $Z.$
\end{itemize}

This corresponds to the \emph{a priori} information on the model, where $Z$ again represents the underlying phenomenon of interest and $\sigma\tau \Ebf$ accounts for an additional measurement noise when $\tau > 0$. Under the Gaussian assumptions, the posterior distribution of the random field $Z$ can be computed explicitly, namely
\begin{equation}
\left(Z \mid \Yobs = \ybf\right),
\end{equation}
where $\Yobs$ denotes the random vector of observations $\left[ Y_1, \dots, Y_n \right]^\top$. This conditional random field is itself a GRF, and its distribution is therefore entirely characterized by its mean and covariance kernel. As explained in~\cite{helbert2009assessment} and recalled here in Proposition~\ref{prop_conv_bayes}, its behavior when $\alpha \to \infty$ is of particular interest.

\begin{prop}[Convergence of the conditional GRF.]\label{prop_conv_bayes}

When $\alpha$, the variance of $A$, tends to infinity, corresponding to a non-informative prior, the conditional field $(Z \mid \Yobs = \ybf)$ converges to a GRF whose mean and covariance coincide with the universal kriging predictor. They are given by 
\begin{equation}\label{distrib_cond}
\left\{
\begin{aligned}
\underset{\alpha \to \infty}{\lim} \esp\left(Z(\sbf)\mid \Yobs = \ybf\right) &= u^\star_\noise(\sbf)\\
\underset{\alpha \to \infty}{\lim} \cov\left(Z(\sbf), Z(\sbf')\mid \Yobs=\ybf\right) &= \sigma^2\left[c(\sbf,\sbf') + K_1(\sbf, \sbf') - k_1(\sbf)^\top\Ktau^{-1}k_1(\sbf')\right]
\end{aligned}
\right.
\end{equation}
with $c(\sbf,\sbf') = \left[\phi_0(\sbf)- \kbf_1(\sbf)^\top\Ktau^{-1}\phibfn\right]^\top\left[\left(\phibfn\right)^\top\Ktau^{-1}\phibfn\right]^{-1}\left[\phi_0(\sbf') - \kbf_1(\sbf')^\top\Ktau^{-1}\phibfn\right].$
\end{prop}
 
The proof is provided in~\ref{proof_bayes}. Note that the \emph{a priori} mean $\mu$ of the random variable $A$ does not affect the results when $\alpha \to \infty$. Therefore, without loss of generality, we set $\mu = 0$ for the remainder of the analysis. The objective of the study is then to investigate this conditional GRF, with a particular focus on generating simulations that reflect both the spline predictor and its associated uncertainty. However, the covariance kernel $K_1$ is linked to the spectrum of the Laplace-Beltrami operator $-\Delta$. For a general compact and connected manifold, these eigenvalues and eigenfunctions are not available in closed form, making standard Gaussian process regression computations intractable. To overcome this difficulty, and as already proposed in~\cite{sire2025spline}, a finite-element approximation of the conditional GRF can be considered.

\section{Finite element approximation of Gaussian Random Fields}\label{sec_fe_approx}

Here, the objective is to construct a discretization of Gaussian random fields on the Riemannian manifold $\manif$ by exploiting the finite element approximation framework, as introduced in~\cite{Lindgren}. We discretize the manifold $\manif$ by introducing a polyhedral approximation obtained by a triangulation $\mathrm{T}$ with $m$ nodes $\cbf_1,\dots,\cbf_m \in \manif$. In particular, we assume (as it is standard with surface finite element methods, cf. \cite{dziuk2013finite}) that the resulting triangulated surface is ``{close enough}'' to the manifold $\manif$ so that we can assume there is a one-to-one mapping between the points of $\manif$, and those of its discretized counterpart. As a result, we use with an abuse of notation the same letter $\manif$ to denote the triangulated surface.

Associated with this triangulation, we define a set of compactly supported basis functions $\{\psi_j\}_{j=1}^m$ on $\manif$. Each basis function $\psi_j$ is piecewise linear over $\mathrm{T}$, equals $1$ at the node $\cbf_j$, and vanishes at all other nodes, corresponding to a first-order finite element discretization. This construction naturally yields the following matrices:
\begin{itemize}
\item the mass matrix $\Mbf \in \mathbb{R}^{m\times m}$, with entries $\left[\Mbf\right]_{ij} = \langle \psi_i, \psi_j \rangle_{\lm}$,
\item the stiffness matrix $\Fbf \in \mathbb{R}^{m\times m}$, with entries $\left[\Fbf\right]_{ij} = \langle \nabla \psi_i, \nabla \psi_j \rangle_{\lm}$,
\item the projection matrix $\Abf_n \in \mathbb{R}^{n\times m}$ mapping the triangulation nodes to the observation points, with entries $\left[\Abf_n\right]_{ij} = \psi_j(\sbf_i)$.
\end{itemize}

The mass matrix $\Mbf$ is symmetric positive definite, and the stiffness matrix $\Fbf$ is symmetric positive semi-definite~\cite{pereira:tel-02499376}. Consequently, one can introduce a matrix square root $\sqrt{\Mbf} \in \mathbb{R}^{m\times m}$ satisfying $\sqrt{\Mbf}\,\sqrt{\Mbf}^{\top} = \Mbf$, which may be obtained, for example, via a Cholesky factorization~\cite{benoit1924note}. Owing to the local support of the basis functions $\{\psi_j\}_{j=1}^m$, both $\Mbf$ and $\Fbf$ are sparse matrices. 

\subsection{Finite element statistical model}

Let us propose a finite-element approximation of the statistical model introduced in Section~\ref{equiv_kriging}. First, it should be reminded that a finite-element approximation relies solely on the values at the triangulation nodes and is then extended over the entire domain by linear combination. For instance, the random field $Z$ modeling the phenomenon will be approximated by
\begin{equation}\label{eq_fe}
Z^{(m)}(\sbf) = \sum_{j=1}^{m} Z_j \psi_j(\sbf) = \Abf(\sbf)^\top \Zbf,~~\sbf\in\manif
\end{equation}

\noindent where for $ 1 \leq j \leq m $, $ Z_j = Z^{(m)}(\cbf_j) $ corresponds to the value of the finite-element approximation $Z^{(m)}$ at the mesh node $\cbf_j$,  $ \Zbf = \left( Z_1,\dots,Z_m \right)^\top$, and $ \Abf(\sbf) = \left( \psi_1(\sbf),\dots,\psi_m(\sbf) \right)^\top$. We therefore suppose that from now on, the observations $\ybf$ from our initial spline problems in Section~\ref{formulation} are noisy realizations of the random field $Z^{(m)}$, i.e. $y_i$ is a realization of
\begin{equation}\label{model_fe}
Y^{(m)}_i = Z^{(m)}(\sbf_i) + \sigma\tau E_i,~~1\leq i\leq n
\end{equation} \noindent where $\Ebf = \left(E_i\right)_{i=1}^{n}$ is a centered Gaussian vector with $\cov(\Ebf) = \Ibf_n$, independent of $\Zbf$.

Finally, following the same approach as in Section~\ref{equiv_kriging}, we assume that the distribution of $\Zbf$, which entirely charecterizes the distribution of $Z^{(m)}$ arises from the following two-stage hierarchical model:
\begin{itemize}
\item $\left(Z^{(m)}(\sbf) \mid A=a\right) = \left(\Abf(\sbf)^\top \Zbf \mid A=a\right)$ is a centered Gaussian vector
\item $A \sim \mathcal{N}(0,\alpha)$ ($\mu = 0$ without loss of generality)
\end{itemize}


The objective is then to define the distribution of the random vector $\left(\Zbf\mid A = a\right),$ so that $\left(Z^{(m)} \mid A = a\right)$ is the finite-element approximation of $\left(Z \mid A=a\right).$ As a reminder, $\forall a \in \mathbb{R},\left(Z\mid A=a\right)$ is a GRF with mean $a\phi_0$ and covariance kernel $\sigma^2K_1.$ The result is presented and justified in~\cite{pereira_desassis_allard_2022,lang2023galerkin}, and is briefly recalled here. 

Let us denote $-\Delta_m$ the Galerkin approximation of $-\Delta$ on the triangulation $\mathrm{T}$, defined explicitly in~\ref{galerk}, and $\left(\lambda_k^{(m)}\right)_{k=0}^{m-1}$ and $\left(\phi_k^{(m)}\right)_{k=0}^{m-1}$ denote its eigenvalues and eigenvectors. Then, to characterize the distribution of $\left(\Zbf \mid A = a\right)$ , we need to introduce the following matrix $\Sbf$ defined as

\begin{equation}\label{eq_S}
\Sbf = \left(\sqrt{\Mbf}\right)^{-1}\Fbf\left(\sqrt{\Mbf}\right)^{-\top},
\end{equation} which is real, symmetric, and positive semi-definite with same eigenvalues as $-\Delta^{(m)}.$~\cite{lang2023galerkin}. Then, using the diagonalization $\Sbf = \Vbf \diag\left(\lambda_0^{(m)},\dots, \lambda_{m-1}^{(m)}\right)\Vbf^\top,$ we have the following result.

\begin{prop}[Distribution of the random vector $\left(\Zbf\mid A=a\right)$.]\label{prop_sigma}
For all $a\in\mathbb{R},$ the random vector $\left(\Zbf \mid A = a\right)$ is a Gaussian vector with 
\begin{itemize}
\item mean $a\phibf$ 
\item covariance matrix
$\sigma^2\sigbf$ with $\sigbf = \left(\sqrt{\Mbf}\right)^{-\top} f(\Sbf) \left(\sqrt{\Mbf}\right)^{-1}$,
\end{itemize}
where $\phibf = \left(\phi_0^{(m)}(\cbf_j)\right)_{j=1}^{m} = \frac{\mathbf{1}_m}{\lVert \sqrt{\Mbf}^\top \mathbf{1}_m \rVert_2}$ (see~\ref{formula_phibf}), $f$ is the function such that $f(\lambda) = \frac{1}{\lambda^2}\mathbf{1}_{\lambda > 0}$ and $f(\Sbf)$ is the matrix function defined from the eigendecomposition of $\Sbf$ as
$$
f(\Sbf) = \Vbf \diag\left(f(\lambda_0^{(m)}),\dots, f(\lambda_{m-1}^{(m)})\right) \Vbf^\top.
$$
\end{prop}

Using proposition~\ref{prop_sigma} and Equation~\ref{eq_fe}, it provides the distribution of the finite-element random field $\left(Z^{(m)}(\sbf) \mid A = a\right).$

\begin{coro}[Distribution of the finite-element approximation of the GRF.]
For all $a\in\mathbb{R},$ $\left(Z^{(m)}(\sbf)\mid A = a\right)$ is a Gaussian random field with mean $a\Abf(\sbf)^\top\phibf$ and covariance kernel $\sigma^2 K_1^{(m)}$ with

$$K_1^{(m)}(\sbf,\sbf') = \Abf(\sbf)^\top\sigbf \Abf(\sbf'),~~~~(\sbf,\sbf')\in\manif^2.$$
\end{coro}

In practice, and as commonly adopted in finite element implementations, the mass matrix $\Mbf$ is replaced by its mass lumping approximation~\cite{quarteroni2016modellistica}. This approximation consists of a diagonal matrix whose entries are given by $\left[\Mbf\right]_{jj} = \langle \psi_j, 1 \rangle$ for $1 \leq j \leq m$. For notational simplicity, we retain the same symbol $\Mbf$ for the mass lumping matrix. This simplification facilitates access to $\left(\sqrt{\Mbf}\right)^{-1}$ and ensures that $\Sbf$ is sparse.

Now that we have characterized the distribution of the discretized field $Z^{(m)}$ used to model the observations $\Yobs^{(m)}=\left[Y_1^{(m)}, \dots, Y_n^{(m)} \right]^\top$ introduced in Equation~\ref{model_fe}, we can move on to the characterization of the posterior distribution $\left(Z^{(m)}\mid \Yobs^{(m)} = \ybf\right)$ when the variance $\alpha$ of $A$ tends to infinity, which will be used to make predictions at new locations.

\subsection{Predictive distribution of the finite element model} 

Following once again the approach in Section~\ref{equiv_kriging}, we wish to use the posterior distribution $\left(Z^{(m)}\mid \Yobs^{(m)} = \ybf\right)$ in the case $\var(A)=\alpha \rightarrow \infty$ as a predictive distribution to recover the observed field at any location. Noting that the vector of observations $\Yobs^{(m)}$ can be written as
\begin{equation}
\Yobs^{(m)} = \Abf_n \Zbf + \sigma\tau \Ebf,
\end{equation}
with  $\Abf_n$ the projection matrix defined in Section~\ref{sec_fe_approx}, we can adapt Proposition~\ref{prop_conv_bayes} to our finite-element statistical model to obtain the following result.

\begin{prop}[Convergence of the finite element conditional GRF]\label{prop_conv_fe}
Let $\alpha =\var(A)$ in the hierarchical model defining $Z^{(m)}$. Then, $\left(Z^{(m)} \mid \Yobs^{(m)} = \ybf \right)$  is a GRF such that
\begin{enumerate}[label=(\roman*)]
\item For any $\sbf\in\mathcal{M}$, the mean $\esp\left(Z^{(m)}(\sbf)\mid \Yobs^{(m)} = \ybf\right)$ satisfies
\begin{align*}
\underset{\alpha \to \infty}{\lim} &\esp\left(Z^{(m)}(\sbf)\mid \Yobs^{(m)} = \ybf\right) = \amnoise \Abf(\sbf)^\top\phibf +\Abf(\sbf)^\top\sigbf \Abf_n^\top \left(\Kbis\right)^{-1}(\ybf - \amnoise\Abf_n\phibf)
\end{align*}

\noindent where $\amnoise = \left(\phibf^\top \Abf_n^\top\left(\Kbis\right)^{-1}\Abf_n\phibf\right)^{-1}\phibf^\top \Abf_n^\top\left(\Kbis\right)^{-1} \ybf$ and $\Kbis~=~\Abf_n\sigbf (\Abf_n)^\top +\tau^2 \Ibf_n$

\item For any $\sbf,\sbf'\in\mathcal{M}$, the covariance kernel $\cov\left(Z(\sbf), Z(\sbf')\mid \Yobs^{(m)}=\ybf\right)$ satisfies
\begin{align*}
\underset{\alpha \to \infty}{\lim} \cov\left(Z(\sbf), Z(\sbf')\mid \Yobs^{(m)}=\ybf\right)
&= \sigma^2\Bigl[
    c^{(m)}(\sbf,\sbf')
    + \Abf(\sbf)^\top \sigbf \Abf(\sbf') \\
&\qquad
    - \Abf(\sbf)^\top \sigbf \Abf_n^\top
      \left(\Kbis \right)^{-1}
      \Abf_n \sigbf \Abf(\sbf')
\Bigr].
\end{align*}
where
$c^{(m)}(\sbf,\sbf') = b(\sbf)\left(\phibf^\top \Abf_n^\top\left(\Kbis\right)^{-1}\Abf_n\phibf\right)^{-1}b(\sbf')$ and\\ $b(\sbf)~=~\Abf(\sbf)^\top\phibf- \Abf(\sbf)^\top\sigbf \Abf_n^\top \left(\Kbis\right)^{-1}\Abf_n\phibf$.
\end{enumerate}
Again, the limiting field defined by these moments remains a GRF.
\end{prop}

The proof is identical to that of Proposition~\ref{prop_conv_bayes} given in~\ref{proof_bayes}, replacing $\phi_0(\sbf)$ with $\phi_0^{(m)}(\sbf) = \Abf(\sbf)^\top \phibf$ and $K_1(\sbf,\sbf')$ with $K_1^{(m)}(\sbf,\sbf') = \Abf(\sbf)^\top \sigbf \Abf(\sbf')$. The inverisibility of the matrix $\Kbis$ is discussed in~\ref{inversibility_kbis}. Using these results, the objective is now to simulate, at the triangulation points, the GRF $\left(Z^{(m)} \mid \Yobs^{(m)} = \ybf\right)$
of Proposition~\ref{prop_conv_fe}, which corresponds to the vector $\left(\Zbf \mid \Yobs^{(m)} = \ybf\right)$ as $\alpha \to \infty$. This will be the purpose of the next section. As an extension, in~\ref{other_simu}, we provide a strategy to generate simulations of $\left(Z^{(m)} \mid \Yobs^{(m)} = \ybf, A = \amnoise\right),$ that is a GRF with the same mean as the limit presented in~Proposition~\ref{prop_conv_bayes}, then matching the finite element approximation of the spline prediction, but whose covariance kernel does not contain the term $c^{(m)}(\sbf,\sbf')$ that refers to the uncertainty on the \emph{a priori} mean, as in the simple kriging framework. 

\section{Predictive posterior simulations}\label{sec_simus}

\subsection{Overview of the approach}\label{general_idea}

In this section, we define $\Zpost \sim \mathcal{N}(\mpost(\ybf), \covpost)$, the limiting predictive distribution of Proposition~\ref{prop_conv_fe}, i.e. we take
\begin{itemize}
    \item $\mpost(\ybf)=\underset{\alpha \to \infty}{\lim} \esp\left(\Zbf\mid \Yobs^{(m)} = \ybf\right)$ as the posterior expectation of $\left(\Zbf\mid \Yobs^{(m)} = \ybf\right)$ in the limit $\var(A)=\alpha\rightarrow\infty$.
    \item $\covpost=\underset{\alpha \to \infty}{\lim}\cov\left(\Zbf\mid \Yobs^{(m)}=\ybf\right)$ as the corresponding covariance kernel, which in particular does not depend on $\ybf$ (see Proposition~\ref{prop_conv_fe}).
\end{itemize}

Our objective is to generate simulations of this Gaussian vector. Two approaches may naturally be considered as possible starting points.

First, within this universal kriging framework, simulations of $\Zpost$ can be derived from prior simulations of $\left(\Zbf \mid A = a\right)$, that has a covariance matrix $\sigma^2\sigbf$ (as established in Proposition~\ref{prop_sigma}). However, computing $\sigbf$ requires the diagonalization of $\Sbf$ (defined in Equation~\ref{eq_S}). Because $\Sbf$ is an $m \times m$ matrix, where $m$ denotes the number of triangulation nodes, this computation becomes computationally intractable for large $m$. To address this, the precision matrix, defined as the inverse of the covariance matrix $\sigbf$, is exploited in practice. The idea is that Gaussian random vectors defined on a triangulation mesh can be viewed as GMRFs, which implies that their precision matrix is sparse. However, direct manipulation of the precision matrix is not feasible in our setting. Since $\operatorname{rank}(\sigbf) = \operatorname{rank}\!\big(f(\Sbf)\big) = m-1$, with $f(\lambda_0^{(m)}) = 0$ and $f(\lambda_k^{(m)}) > 0$ for $1 \leq k \leq m-1$, the random vector $\left(\Zbf \mid A = a\right)$ is an intrinsic GMRF, rendering classical computational methods involving the precision matrix inapplicable.

Alternatively, one may approximate the $\alpha \to \infty$ limit by simulating $\left(\Zbf \mid \Yobs^{(m)} = \ybf\right)$ for a sufficiently large $\alpha$. These simulations leverage the invertibility of $\sigalp$, the covariance matrix of $\Zbf,$ defined as $\sigalp = \sigma^2\sigbf + \alpha \phibf \phibf^\top$. The computations leading to this covariance matrix are identical to those detailed in~\ref{app_distrib_z}. Its inverse, the precision matrix of $\Zbf$, is defined as: 
\begin{equation}\label{eq_qalp}
 \qalp = \sigalp^{-1} = \sigma^{-2}\sqrt{\Mbf}\Sbf^2 \sqrt{\Mbf}^\top + \frac{1}{\alp}\left(\Mbf\phibf\right)\left(\Mbf\phibf\right)^\top.
 \end{equation} 
The structure of $\qalp$ is well-suited for sparse computations. Specifically, one can compute the Cholesky factorization of the sparse matrix $\sqrt{\Mbf}\Sbf^2 \sqrt{\Mbf}^\top$ and subsequently obtain the Cholesky factorization of $\qalp$ via a rank-one update. However, this strategy leads to numerical difficulties that depend on the value of $\alp$ and relies on the approximation of a limit, which prevents exact computation. \\~\

The approach proposed here involves leveraging the matrix $\qalp$ for a fixed $\alpha$ to derive exact simulations of $\Zpost$, thereby avoiding the need for a limiting approximation. This process is decomposed into the following steps:

\begin{enumerate}
	\item Generate simulations of the intrinsic GMRF $\left(\Zbf \mid A = a\right)$ using $\qalp$ (Section~\ref{simu_intrinsic}).
    \item Compute the posterior expectation $\mpost(\ybf)$ using $\qalp$, without limit estimation (Section~\ref{post_exp})
    \item Construct exact simulations of $\Zpost$ by combining the results from the previous two steps (Section~\ref{condi_simu})
\end{enumerate}

The overall procedure is summarized in Section~\ref{overall_algo}.

\subsection{Simulating an intrinsic GMRF}\label{simu_intrinsic}

Here, we detail the procedure to obtain simulations, for a given $a\in \mathbb{R},$ of the intrinsic GMRF $$\Zprior \eqdef \left(\Zbf \mid A = a\right) \sim \mathcal{N}(a\phibf, \sigma^2\sigbf),$$ using the GMRF $\Zbf\sim\mathcal{N}(\veczerm, \qalp^{-1})$. The work of~\cite{simpson2006sampling} can be adapted here and leads to the following result, which is valid for all $\alpha > 0.$

\begin{prop}[Simulation of intrinsic GMRF]\label{prop_simu_intrinsic}
Considering $\Zbf \sim \mathcal{N}(\veczerm, \sigalp)$, then 
$$\left(a\phibf+\Zbf \mid \left[\left(\Mbf\phibf\right)^\top\Zbf = 0 \right]\right) = a\phibf+ \Zbf - \phibf\left(\Mbf\phibf\right)^\top\Zbf \sim \mathcal{N}(a\phibf, \sigma^2\sigbf).$$
\end{prop}
\begin{proof}
The vector
\[
\Zbf - \phibf\left(\Mbf\phibf\right)^\top \Zbf
= \left(\Ibf_m - \phibf\left(\Mbf\phibf\right)^\top\right)\Zbf
\]
is a centered Gaussian vector, since it is a linear transformation of the centered Gaussian vector $\Zbf$.
Its covariance is

\begin{align*}
&\left(\Ibf_m - \phibf\left(\Mbf\phibf\right)^\top\right)\cov(\Zbf)\left(\Ibf_m - \phibf\left(\Mbf\phibf\right)^\top\right)^\top  \\ 
=& \left(\Ibf_m - \phibf\left(\Mbf\phibf\right)^\top\right)\left(\sigma^2\sigbf +\alpha\phibf\phibf^\top\right)\left(\Ibf_m - \Mbf\phibf\phibf^\top\right).
\end{align*}
As shown in~\ref{inversibility_kbis}, $\sigbf\Mbf\phibf = \veczerm.$ Then, 
\begin{align*}
&\left(\Ibf_m - \phibf\left(\Mbf\phibf\right)^\top\right)\left(\sigma^2\sigbf +\alpha\phibf\phibf^\top\right)\left(\Ibf_m - \Mbf\phibf\phibf^\top\right) \\
= &~\sigma^2\sigbf + \alpha\phibf\phibf^\top -  2\alpha\phibf\left(\Mbf\phibf\right)^\top \phibf\phibf^\top + \alpha \phibf\left(\Mbf\phibf\right)^\top\phibf\phibf^\top\Mbf\phibf\\
= &~\sigma^2\sigbf + \alpha\phibf\phibf^\top -  2\alpha\phibf\phibf^\top + \alpha \phibf\phibf^\top\text{ using } \phibf^\top\Mbf\phibf = 1 \text{ see~\ref{formula_phibf}}.\\
= &~\sigma^2\sigbf
\end{align*}
\end{proof}

Since the simulation of $\Zbf \sim \mathcal{N}(\veczerm,\sigalp)$ can be carried out by solving a linear system involving $\Lbf_\alp$ the Cholesky factor of $\qalp$, and a vector of independent standard Gaussian components, this yields the following algorithm for simulating $\Zprior.$

\begin{algorithm}[H]
\caption{Simulation of $\Zprior\sim\mathcal{N}(a\phibf, \sigma^2\sigbf)$}
\begin{algorithmic}[1]
\Require $\Mbf, \Fbf,\phibf,\sigma,\alpha,a$ and a vector $\gaussbfm$ of $m$ independent standard Gaussian components
\State $\Sbf \gets \sqrt{\Mbf}^{-1} \, \Fbf \, \sqrt{\Mbf}^{-\top}$
\State $\qbf \gets \sqrt{\Mbf} \, \Sbf^2 \, \sqrt{\Mbf}^{\top}$
\State Compute $\Lbf_\alp$ the Cholesky factor $\qalp$ using sparsity of $\sigma^{-2}\qbf$ and a rank-one update 
\State Solve the linear system $\Lbf_\alp^\top \zbf = \gaussbfm$
\State \Return $\zbf_a = a\phibf+ \zbf - \phibf\left(\Mbf\phibf\right)^\top \zbf$
\end{algorithmic}\label{algo_simuprior}
\end{algorithm}

\subsection{Computing the posterior expectation}\label{post_exp}

The objective here is to compute $\mpost(\ybf)$. The work of~\cite{sire2025spline} establishes an analytical relationship between $$\mpostalp(\ybf)\eqdef \esp\!\left(\Zbf\mid \Yobs^{(m)}=\ybf\right) = \qalp^{-1} \Abf_n^\top \left[\Abf_n\qalp^{-1} \Abf_n^\top + \sigma^2\noise^2\Ibf_n \right]^{-1}\ybf,$$ and $\mpost(\ybf)=\underset{\alpha\to\infty}{\lim}\mpostalp(\ybf).$

\begin{prop}[Computation of the posterior expectation.]\label{compute_mean}
The posterior expectation $\mpost(\ybf)$ can be computed using $\mpostalp(\ybf)$ and $\mpostalp\left(\Abf_n\phibf\right),$ using the relation 
$$\mpost(\ybf)
= \mpostalp(\ybf) 
+ \Bigg(
   \frac{\phibf - \hh\!\left[\mpostalp\!\left(\Abf_n \phibf\right)\right]}
        {\left(\Mbf\phibf\right)^\top \mpostalp\!\left(\Abf_n \phibf\right)}
   - \phibf 
   + \hh\!\left[\mpostalp\!\left(\Abf_n \phibf\right)\right]
  \Bigg) 
  \left(\Mbf\phibf\right)^\top \mpostalp(\ybf),$$

\noindent where $\hh\!\left[\mbf\right] 
= \frac{\mbf - \phibf\left(\Mbf\phibf\right)^\top \mbf}{1 - \left(\Mbf\phibf\right)^\top \mbf},~~\forall \mbf\in\mathbb{R}^m.$
\end{prop}

Using Proposition~\ref{compute_mean} which is valid for all $\alpha > 0$, the posterior expectation $\mpost(\ybf)$ can be computed for any $\ybf\in \mathbb{R}^n,$, provided that $\mpostalp(\ybf)$ is tractable for any $\ybf\in \mathbb{R}^n.$ For the computation of $\mpostalp(\ybf)$, two scenario are distinguished.

\begin{enumerate}
\item The first scenario corresponds to the study of interpolating splines, i.e., setting the noise $\tau = 0$. To facilitate the computations in this situation, we require that the observation points $\mathcal{S} = (\sbf_i)_{i=1}^n$ are included among the triangulation nodes $\cbf_1,\dots,\cbf_m$. More precisely, let $I = \{i_1,\dots, i_n\}$ denote the set of indices such that $s_{k} = c_{i_k}$ for $1 \leq k \leq n$. This requirement is not restrictive in practice, as the triangulation can typically be designed to include the observation points, while simplifying the computations. Indeed, with this hypothesis, we have
$$
(\Abf_n)_{k,j} =
\begin{cases}
1 & \text{if } j = i_k, \\
0 & \text{otherwise},
\end{cases}
$$
for $1\leq k \leq n$ and $1\leq j \leq m,$ leading to $\Abf_n 
\qalp^{-1}\Abf_n^\top = \left[\qalp^{-1}\right]_{I,I},$ by denoting \(\left[\qalp\right]_{I_1,I_2}\) the submatrix of \(\qalp\) consisting of the rows indexed by \(I_1\) and the columns indexed by \(I_2\) for any index sets \(I_1, I_2\). Then, we have the simplified results

$$\begin{cases}
\big[\mpostalp(\ybf)\big]_{\bar{I}} &= \big[\tilde{\qbf}_\alpha\big]_{\bar{I},\bar{I}}^{-1}\big[\tilde{\qbf}_\alpha\big]_{\bar{I},I}\ybf \\
\big[\mpostalp(\ybf)\big]_{I} &= \ybf,
\end{cases}$$

\noindent and the vector $\mpostalp(\ybf)$ can be computed with the following procedure.

\begin{algorithm}[H]
\caption{Computation of $\mpostalp(\ybf)$ in the first scenario}
\begin{algorithmic}[1]
\Require $\Mbf, \Fbf, \phibf, I,\sigma,\alpha,\ybf$
\State $\Sbf \gets \sqrt{\Mbf}^{-1} \, \Fbf \, \sqrt{\Mbf}^{-\top}$
\State $\qbf \gets \sqrt{\Mbf} \, \Sbf^2 \, \sqrt{\Mbf}^{\top}$
\State $\phibf \gets \frac{\vecunm}{\lVert \sqrt{\Mbf}^\top \vecunm \rVert_2}$
\State Compute the Cholesky decomposition of $\left[\tilde{\qbf}_\alpha\right]_{\bar{I},\bar{I}} = \sigma^{-2}\left[\qbf\right]_{\bar{I},\bar{I}} + \frac{1}{\alpha}\left[\Mbf\phibf\right]_{\bar{I}}\left[\Mbf\phibf\right]_{\bar{I}}^\top$ using sparsity of $\sigma^{-2}\left[\qbf\right]_{\bar{I},\bar{I}}$ and a rank-one update
\State Solve $\left[\tilde{\qbf}_\alpha\right]_{\bar{I},\bar{I}}\ombf = \left[\tilde{\qbf}_\alpha\right]_{\bar{I},I}\ybf$ using sparse linear systems
\State \Return $\begin{bmatrix}
[\mpostalp(\ybf)]_{\bar{I}} \\
[\mpostalp(\ybf)]_{I}\end{bmatrix} = \begin{bmatrix}
 \ombf\\ \ybf
\end{bmatrix}$
\end{algorithmic}\label{algo_malp1}
\end{algorithm}

\item The second scenario corresponds to the computation of smoothing splines, i.e., setting the noise $\noise > 0$. In this case, the assumption that the observation points belong to the triangulation nodes is no longer required to obtain computational simplifications. Indeed, one has
$$
\tilde{\qbf}_\alp^{-1} \Abf_n^\top\!\left( \Abf_n \tilde{\qbf}_\alp^{-1} \Abf_n^\top + \sigma^2\noise^2\Ibf_n\right)^{-1}\ybf
=
\left(\sigma^2\noise^2\tilde{\qbf}_\alp+\Abf_n^\top\Abf_n\right)^{-1}\Abf_n^\top\ybf$$
by the Woodbury formula. Using the fact that $\Abf_n^\top\Abf_n$ is sparse by construction of the basis functions $(\psi_j)_{j=1}^{m}$, the decomposition of $\sigma^2\noise^2\tilde{\qbf}_\alp+\Abf_n^\top\Abf_n$ can be efficiently obtained by expressing this matrix as the sum of a sparse matrix and a rank-one update. Algorithm~\ref{algo_malp2} then provides the computation of $\mpostalp(\ybf)$ for $\noise > 0$.

\begin{algorithm}[H]
\caption{Computation of $\mpostalp(\ybf)$ in the second scenario}\label{algo_malp2}
\begin{algorithmic}[1] 
\Require $\Mbf, \Fbf, \phibf, \alpha, \Abf_n, \ybf, \noise > 0, \sigma>0$
\State $\Sbf \gets \sqrt{\Mbf}^{-1} \, \Fbf \, \sqrt{\Mbf}^{-\top}$
\State $\qbf \gets \sqrt{\Mbf} \, \Sbf^2 \, \sqrt{\Mbf}^{\top}$
\State $\phibf \gets \frac{\vecunm}{\lVert \sqrt{\Mbf}^\top \vecunm \rVert_2}$
\State Compute the Cholesky of $\sigma^2\noise^2\tilde{\qbf}_\alp + \Abf_n^\top\Abf_n = \noise^2\qbf + \Abf_n^\top\Abf_n + \frac{\sigma^2\noise^2}{\alp}\left(\Mbf\phibf\right)\left(\Mbf\phibf\right)^\top$ using sparsity of $\noise^2\qbf+ \Abf_n^\top\Abf_n$ and a rank-one update
\State Solve $\left(\sigma^2\noise^2\tilde{\qbf}_\alp + \Abf_n^\top\Abf_n\right)\ombf = \Abf_n^\top\ybf$ using sparse linear systems
\State \Return $\mpostalp(\ybf) = \ombf$
\end{algorithmic}
\end{algorithm}
\end{enumerate}

\subsection{Generating the conditional simulations}\label{condi_simu}

Our approach exploits the prior simulations of $\left(\Zprior = \Zbf \mid A = a\right)$ (Section~\ref{simu_intrinsic}) and the calculation of the posterior expectation (Section~\ref{post_exp}) to obtain simulations of $\Zpost$. To this end, we use the following result, which specifically addresses the simulation of residuals~\cite{chiles}.

\begin{prop}[Simulation of the residuals.]\label{simu_res}
For all $a\in \mathbb{R}$, let $\Zprior \sim \mathcal{N}(a \phibf, \sigma^2\sigbf),$ and $\Ebf \sim \mathcal{N}(0, \Ibf_{n}).$ Then, the vector of residuals
$$\Zres = \mpost\left(\Abf_n \Zprior + \sigma\tau \Ebf\right) - \Zprior$$
is a centered Gaussian vector with covariance matrix $\covpost.$
\end{prop}
\begin{proof}
The objective is to compute 

\begin{equation}
\left\{
\begin{aligned}
\esp(\Zres) &= \esp\left(\mpost\left(\Abf_n \Zprior + \sigma\tau \Ebf\right)\right) - a\phibf\\
\cov(\Zres) &= \cov\left[\mpost\left(\Abf_n \Zprior + \sigma\tau \Ebf\right)\right]- 2\cov\left[\mpost\left(\Abf_n \Zprior + \sigma\tau \Ebf\right), \Zprior\right]+ \cov\left[\Zprior\right]
\end{aligned}
\right.
\end{equation}
using 
\begin{align*}
&\cov\left(\Abf_n \Zprior + \sigma\tau \Ebf\right) = \sigma^2\Kbis\\
&\cov\left[\Zprior\right] = \sigma^2\sigbf.
\end{align*}\medskip

\noindent\textbf{Step 1:} \textbf{Compute $\mpost\left(\Abf_n \Zprior + \sigma\tau \Ebf\right).$} Using Proposition~\ref{prop_conv_fe}, have $$\mpost\left(\Abf_n \Zprior + \sigma\tau \Ebf\right) = \ambis \phibf +\sigbf \Abf_n^\top \left(\Kbis\right)^{-1}(\Abf_n \Zprior + \sigma\tau\Ebf - \ambis\Abf_n\phibf)$$

\noindent with $\ambis = \left(\phibf^\top \Abf_n^\top\left(\Kbis\right)^{-1}\Abf_n\phibf\right)^{-1}\phibf^\top \Abf_n^\top\left(\Kbis\right)^{-1}\left(\Abf_n \Zprior + \sigma\tau \Ebf\right).$

Let us denote $\bbf = \phibf - \sigbf\Abf_n^\top\left(\Kbis\right)^{-1}\Abf_n\phibf$ and $\beta =\phibf^\top \Abf_n^\top\left(\Kbis\right)^{-1}\Abf_n\phibf.$ Then,  
\begin{equation}\label{eq_mpost}
\begin{aligned}
\mpost\left(\Abf_n \Zprior + \sigma\tau \Ebf\right) &= \bbf \ambis + \sigbf \Abf_n^\top \left(\Kbis\right)^{-1}\left[\Abf_n \Zprior + \sigma\tau \Ebf\right] \\
&= \beta^{-1}\bbf\phibf^\top\Abf_n^\top\left(\Kbis\right)^{-1}\left[\Abf_n \Zprior + \sigma\tau \Ebf\right] + \sigbf \Abf_n^\top \left(\Kbis\right)^{-1}\left[\Abf_n \Zprior + \sigma\tau \Ebf\right]\\
&= \left[\beta^{-1}\bbf\phibf^\top\Abf_n^\top+\sigbf \Abf_n^\top \right]\left(\Kbis\right)^{-1}\left[\Abf_n \Zprior + \sigma\tau \Ebf\right].
\end{aligned}
\end{equation}

\noindent\textbf{Step 2:} \textbf{Compute $\esp\left[\mpost\left(\Abf_n \Zprior + \sigma\tau \Ebf\right)\right]$ and $\cov\left[\mpost\left(\Abf_n \Zprior + \sigma\tau \Ebf\right)\right]$}. From step 1, we have

\begin{align*}
\esp\left[\mpost\left(\Abf_n \Zprior + \sigma\tau \Ebf\right)\right] &= a\left[\beta^{-1}\bbf\phibf^\top\Abf_n^\top+\sigbf \Abf_n^\top \right]\left(\Kbis\right)^{-1}\Abf_n \phibf\\
&= a\left(\bbf + \sigbf\Abf_n^\top\left(\Kbis\right)^{-1}\Abf_n\phibf\right)\\
&= a\phibf \text{ from the definition of }\bbf.
\end{align*}

For the covariance, we obtain

\begin{align*}
\cov\left[\mpost\left(\Abf_n \Zprior + \sigma\tau \Ebf\right)\right] &= \sigma^2\left[\beta^{-1}\bbf\phibf^\top\Abf_n^\top+\sigbf \Abf_n^\top \right]\left(\Kbis\right)^{-1}\left[\beta^{-1}\bbf\phibf^\top\Abf_n^\top+\sigbf \Abf_n^\top \right]^\top\\
& = \sigma^2\left[\beta^{-1}\bbf^\top\bbf + \left[2\beta^{-1}\bbf\phibf^\top\Abf_n^\top + \sigbf \Abf_n^\top\right]\left(\Kbis\right)^{-1}\Abf_n\sigbf\right].
\end{align*}\medskip

\noindent\textbf{Step 3:} \textbf{Compute $\cov\left[\mpost\left(\Abf_n \Zprior + \sigma\tau \Ebf\right), \Zprior\right].$} Using Equation~\ref{eq_mpost}, we have 
\begin{align*}
\cov\left[\mpost\left(\Abf_n \Zprior + \sigma\tau \Ebf\right), \Zprior\right] &= \left[\beta^{-1}\bbf\phibf^\top\Abf_n^\top+\sigbf \Abf_n^\top \right]\left(\Kbis\right)^{-1}\cov\left[\Abf_n \Zprior + \sigma\tau \Ebf,\Zprior\right] \\
&= \sigma^2\left[\beta^{-1}\bbf\phibf^\top\Abf_n^\top+\sigbf \Abf_n^\top \right]\left(\Kbis\right)^{-1}\Abf_n \sigbf
\end{align*}

\noindent\textbf{Step 4:} \textbf{Gathering the results.} It comes $\esp(\Zres) = 0$, and using $\cov(\Zprior) = \sigma^2\sigbf,$ it leads to 
\begin{align*}
\cov(\Zres) &= \sigma^2\left[\bbf^\top\beta^{-1}\bbf -\sigbf\Abf_n^\top\left(\Kbis\right)^{-1}\Abf_n \sigbf +\sigbf\right] \\
&= \underset{\alpha \to \infty}{\lim} \cov\left(\Zbf\mid \Yobs^{(m)}=\ybf\right) \text{ using Proposition~\ref{prop_conv_fe}}\\
&= \covpost
\end{align*}
\end{proof}

Then, the strategy adopted here to obtain a simulation of $\Zpost \sim \mathcal{N}(\mpost(\ybf),\covpost)$ is to simulate $\Zprior\sim\mathcal{N}(a\phibf, \sigma\sigbf)$ for a given $a$ and then return
\begin{equation}\label{eq_zpost}
\Zpost = \mpost(\ybf) + \mpost\left(\Abf_n \Zprior + \sigma\tau \Ebf\right) - \Zprior.
\end{equation}
As the result does not depend on $a,$ we work with $a=0$, without loss of generality. It leads to the following overall procedure, detailed in Section~\ref{overall_algo}.

\subsection{Overall simulation algorithm}\label{overall_algo}

Following the procedure outlined in the previous sections, Algorithm~\ref{algo_overall} presents the method for generating spline simulations, with or without observation noise.

\begin{algorithm}[H]
\caption{Generation of splines simulations}\label{algo_overall}
\begin{algorithmic}[1]
\Require $\Mbf, \Fbf,\phibf,\alpha, \noise,\sigma,\ybf$, a vector $\gaussbfm$ of $m$ independent standard Gaussian components
\State Generate a simulation $\zbf_0$, centered with covariance matrix $\sigma^2\sigbf$, using Algorithm~\ref{algo_simuprior}
\If{$\noise = 0$} 
\State \textbf{Require :} $I$
\State Compute $\mpostalp(\ybf)$, $\mpostalp(\left[\zbf_0\right]_{I})$ and $\mpostalp\left(\Abf_n\phibf\right)$ using Algorithm~\ref{algo_malp1}
\State Compute $\mpost(\ybf)$ and $\mpost(\left[\zbf_0\right]_{I})$ using Proposition~\ref{compute_mean} 
\Else 
\State \textbf{Require :} $\Abf_n$ and $\gaussbfn$ a vector of $n$ independent standard Gaussian components
\State Compute $\mpostalp(\ybf)$, $\mpostalp(\Abf_n\zbf_0+\sigma\noise\gaussbfn)$ and $\mpostalp\left(\Abf_n\phibf\right)$ using Algorithm~\ref{algo_malp2}
\State Compute $\mpost(\ybf)$ and $\mpost(\Abf_n\zbf_0+\sigma\noise\gaussbfn)$ using Proposition~\ref{compute_mean} 
\EndIf
\State \Return $\zpost = \mpost(\ybf) +  \mpost(\Abf_n\zbf_0) - \zbf_0$
\end{algorithmic}
\end{algorithm}

Note that these results are theoretically valid for any $\alpha > 0$. However, in practice, this parameter cannot be chosen too large or too small in order to avoid numerical instabilities. As detailed in~\cite{sire2025spline} and adapted here to account for the additional parameter $\sigma$, it is important to ensure that
\[
\frac{\sigma}{\sqrt{\alpha}} \in [\lambda_1^{(m)}, \, \lambda_{m-1}^{(m)}],
\]
to control the condition number of the matrices involved in the Cholesky decomposition. The power-iteration method can be used to estimate these eigenvalues~\cite{mises1929praktische}.

\section{Anisotropy and parameter estimation} \label{sec_likeli}

\subsection{Anisotropic splines}

These spline-based simulations are particularly valuable as they quantify uncertainty alongside the predicted field. However, the covariance kernels introduced above, $K_1$ and its finite-element approximation $K_1^{(m)}$, are isotropic, implying that the modeled phenomenon must exhibit identical behavior in all spatial directions. A natural way to incorporate directional anisotropy, while preserving the exact modeling framework presented here, is to modify the parametrization of the manifold $\left(\manif, g\right)$ via its Riemannian metric $g$. This strategy is described in detail in \cite{pereira_desassis_allard_2022} and is briefly summarized below.

Considering that $\manif$ is embedded in $\mathbb{R}^d$, for any point $\sbf \in \manif$ we consider a coordinate chart $(U, x)$ such that $\sbf \in U$, where $U$ is an open subset of $\manif$ and $x : U \to \mathbb{R}^d$ is a homeomorphism onto an open subset of $\mathbb{R}^d$. The Riemannian metric expressed in local coordinates is given by
\begin{equation}
g_{\sbf}(\ubf_{\sbf}, \vbf_{\sbf}) 
= \left(\ubf_{\sbf}^{x}\right)^{\top} \Gbf^{x}(\sbf)\, \vbf_{\sbf}^{x},
\end{equation}
where $\ubf_{\sbf}^{x}$ and $\vbf_{\sbf}^{x}$ denote the coordinate representations of the tangent vectors $\ubf_{\sbf}, \vbf_{\sbf} \in T_{\sbf}\manif$ in the chart $(U,x)$, and $\Gbf^{x}(\sbf)$ is the matrix representation of the Riemannian metric at $\sbf$ in these local coordinates. The matrix $\Gbf^{x}(\sbf)$ is symmetric positive definite and admits the decomposition
\[
\Gbf^x(\sbf) = \Rbf^x(\sbf)\, \Dbf^x(\sbf)^2 \left(\Rbf^x(\sbf)\right)^\top,
\]
which can be interpreted as a local deformation operator. For $d \in \{2,3\}$, the matrix $\Rbf^x(\sbf)$ corresponds to a rotation, while $\Dbf^x(\sbf)$ is a diagonal matrix containing scaling factors. Under this interpretation, the Riemannian metric induces a local deformation of the manifold by first applying a rotation and then an anisotropic scaling that emphasizes preferred directions. These rotations and scalings are parametrized and can be specified locally on each element of the triangulation. The corresponding parameters are collected in a vector $\param$, which modifies the inner product on $\lm$ and, consequently, the matrices $\Mbf$ and $\Sbf$.

\subsection{Parameter estimation}\label{estim_param}

We then have unknown parameters $\param$ defining the anisotropies, and the variance $\sigma^2.$ A standard approach is to rely on the universal kriging statistical model for the observations and to maximize the associated log-likelihood. It corresponds to the log-density of the random vector $\left(\Yobs^{(m)} \mid A = a\right)$, which is maximized with respect to $\left(a, \sigma^2, \param\right)$

$$\llik\left(a, \sigma, \param\right) = -\frac{1}{2} \left(n\log(2\pi)+\log \lvert \sigma^2\Kbis \rvert + \sigma^{-2}\left(\ybf - a\Abf_n \phibf\right)^\top \left(\Kbis\right)^{-1}\left(\ybf - a \Abf_n \phibf\right) \right),$$

\noindent where $\param$ impacts both $\Kbf$ and $\phibf.$ The optimization leads to 
\begin{equation}
\left\{
\begin{aligned}
a^\star &= \amnoise\\
\sigma^\star &= \left[\frac{1}{n}\left(\ybf - a\Abf_n\phibf\right)^\top\left(\Kbis\right)^{-1}\left(\ybf - a \Abf_n\phibf\right)\right]^{1/2}
\end{aligned}
\right.
\end{equation}
and we end up maximizing in $\param$ the concentrated log-likelihood~\cite{jones2001taxonomy}
$$\llikc(\param) = -n \log\left(\sigma^\star\right) - \frac{1}{2}\log \lvert\Kbis\rvert.$$

The computation of $\sigma^\star$ and $\lvert \Kbis \rvert$ is detailed in~\cite{sire2025spline}. Note that, in the smoothing spline setting, the noise variance $\noise > 0$ can be estimated jointly with $\param$.

\section{Application test cases}\label{sec_appli}

Here, we present two application test cases involving the prediction of smooth phenomena from a limited number of observations. The first application concerns temperature prediction on the Earth, modeled as a spherical surface. The second considers an analytical function defined on a cylindrical surface. In both cases, anisotropies are estimated in local coordinate charts, namely spherical coordinates for the sphere and cylindrical coordinates for the cylinder. For both test cases, the locations of the observation points are fully controlled, so we consider the first scenario and focus on the interpolation problem. The results obtained in the second scenario are similar and are provided in the supplementary material~\cite{charliesire_2026}.

\subsection{Prediction on Earth surface}\label{earth}

Here, we consider the daily analysed Sea Surface Temperature (SST) field for December 26$^\text{th}$, 2023, obtained from the Copernicus satellite SST dataset~\cite{c3s_sst_2025}. It corresponds to a gap-free, gridded estimate of sea surface temperature at approximately $20~\mathrm{cm}$ depth, obtained through optimal interpolation combining observations from multiple satellite sensors. The SST field evolves smoothly over the oceans and exhibits strong anisotropy, as latitude is expected to have a major influence on temperature values. 

We collect the data on a latitude-longitude grid with a $1.5^\circ \times 1.5^\circ$ resolution, which defines the triangulation mesh $(\cbf_j)_{j=1}^{n}$. Obviously, the data points are restricted to oceanic regions only. For the prediction study, we use only $n = 20$ observation points $\ybf = \left(y(\sbf_i)\right)_{i=1}^{n}$, in order to illustrate the effectiveness of splines for reconstructing a smooth phenomenon from a limited number of observations. The observation locations are selected using a Maximin Latin Hypercube Sampling design (LHS,~\cite{kenny2000algorithmic}). In the first scenario, the observation points must belong to the mesh nodes $(\cbf_j)_{j=1}^{n}$. The final design is therefore obtained by selecting the mesh nodes closest to the points generated by the LHS. The remaining mesh nodes, for which the data values are known, are then used as a validation dataset. For the anisotropic model, two parameters were considered: the rotation angle in spherical coordinates and the ratio between the two scaling parameters in this coordinate chart. These parameters were estimated by maximum likelihood, as described in Section~\ref{estim_param}. 
\begin{figure}[h]
\centering
\includegraphics[width=\textwidth]{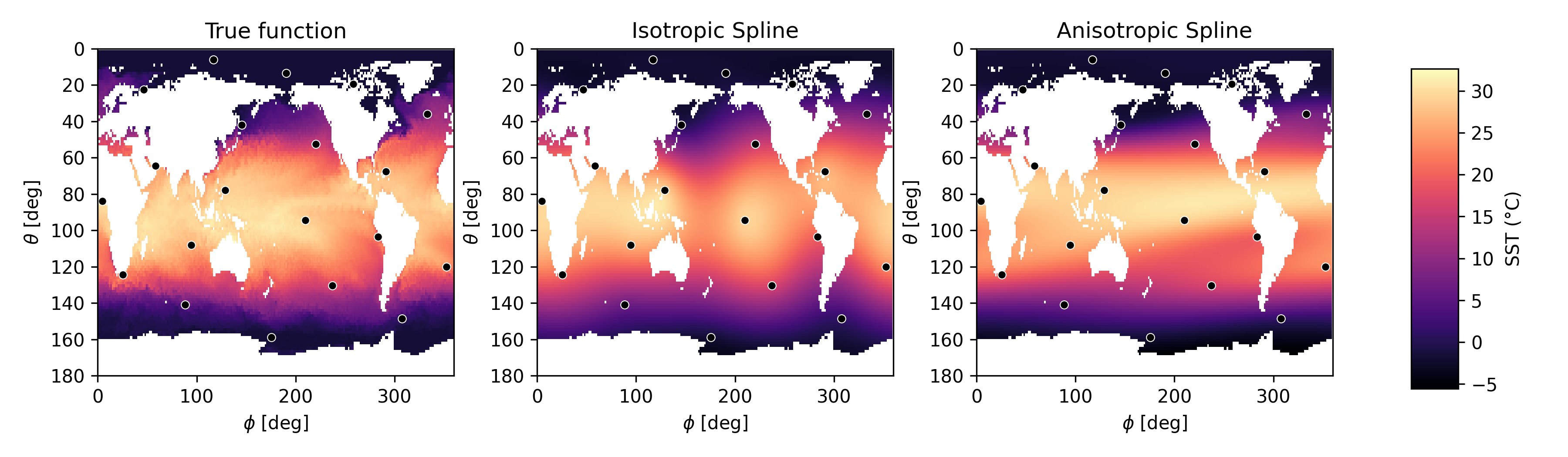}
\caption{Predictions of the SST ($^\circ C$) on Earth, illustrated in 2D using the spherical coordinates $(\theta,\phi).$ $n=20$ observation points are shown as black dots. 
Left: true values. 
Middle: isotropic splines. Right : Anisotropic splines. }
\label{pred_earth}
\end{figure} 

Figure~\ref{pred_earth} displays the prediction mean obtained on the sphere over the oceanic regions using the isotropic and the anisotropic spline. It clearly highlights the influence of anisotropy along the longitudinal direction, which substantially improves the prediction, as illustrated by the distribution of prediction errors in Figure~\ref{errors_earth} and the associated Root Mean Square Error (RMSE). Figure~\ref{uncert_earth} displays the uncertainty associated with the prediction, obtained from $500$ simulations, along the line of longitude closest to $\phi = 170^\circ$ in the triangulation mesh, crossing the Pacific Ocean. Once again, the anisotropic splines better capture the behavior of the phenomenon. However, this figure also emphasizes the importance of properly assessing uncertainty. In the isotropic case, although the prediction mean is less accurate, the associated uncertainty remains consistent. For example, when computing the $95\%$ prediction interval over all points of the triangulation mesh, the true SST value lies within this interval for more than $98\%$ of the points in the isotropic setting. 

However, such coverage is not a relevant performance criterion, since an arbitrarily large prediction variance would ensure that all true values fall within a given confidence interval. As discussed in~\cite{gneiting2007strictly}, a relevant metric that promotes a proper balance between predictive accuracy and uncertainty quantification is obtained by considering, at a prediction point $\sbf$, the score
\begin{equation}\label{eq_score}
S(\sbf) = -\left(\frac{\yhat(\sbf) - \ytrue(\sbf)}{\sigmahat^2(\sbf)}\right)^2 
- \log\left(\sigmahat^2(\sbf)\right),
\end{equation}
which is directly related to the log-likelihood of the Gaussian predictive distribution at $\sbf$, with predictive mean $\yhat(\sbf)$ and predictive variance $\sigmahat^2(\sbf)$. This score is computed at all points of the triangulation mesh, and its distribution is shown in Figure~\ref{scores_earth}, with higher scores indicating better predictive performance. The results once again confirm the benefit of incorporating anisotropy in the spline model for this phenomenon. 

\begin{figure}
    \centering
    
    \begin{subfigure}[t]{0.48\textwidth}
        \centering
        \includegraphics[width=\linewidth]{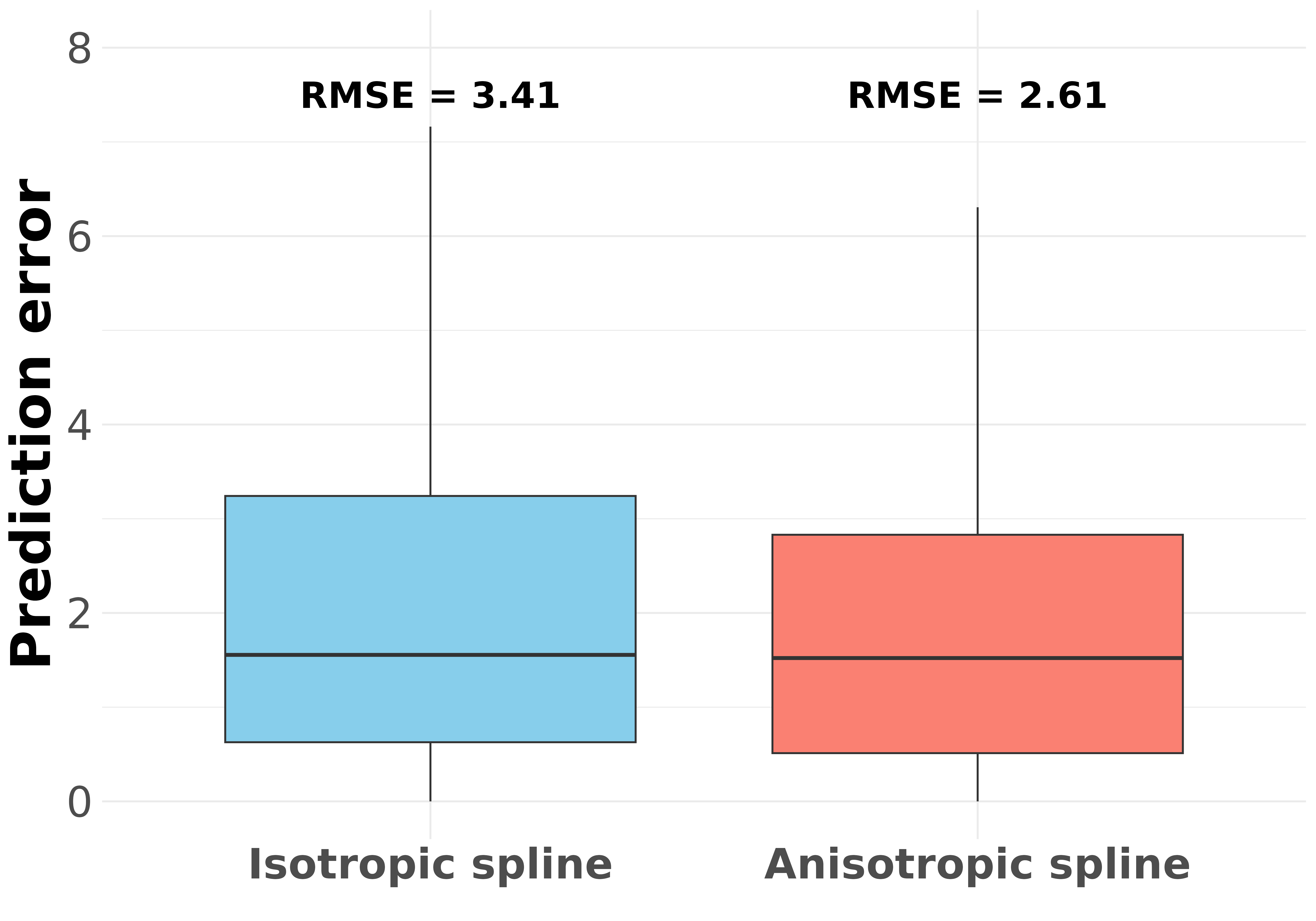}
\caption{Distribution of the prediction errors and associated RMSE in the isotropic and the anisotropic case.}
        \label{errors_earth}
    \end{subfigure}
    \hfill
    \begin{subfigure}[t]{0.48\textwidth}
        \centering
        \includegraphics[width=\linewidth]{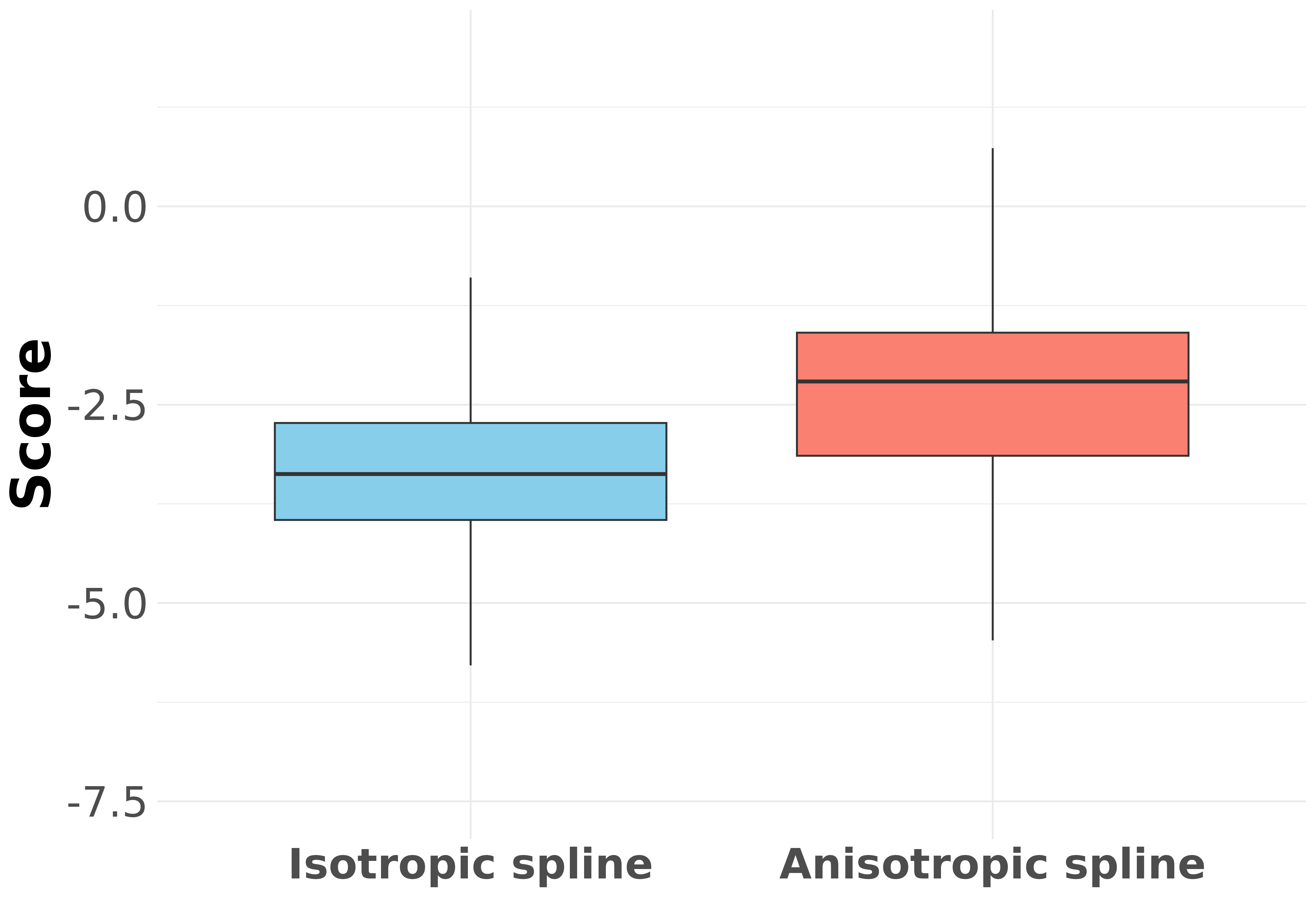}
        \caption{Distribution of the scores for the isotropic and the anisotropic prediction.}
        \label{scores_earth}
    \end{subfigure}
\caption{Comparison of the prediction errors and scores for the isotropic and anisotropic SST predictions on the Earth.}
\end{figure}

\begin{figure}
    \centering
    \begin{subfigure}[t]{0.48\textwidth}
        \centering
        \includegraphics[width=\linewidth]{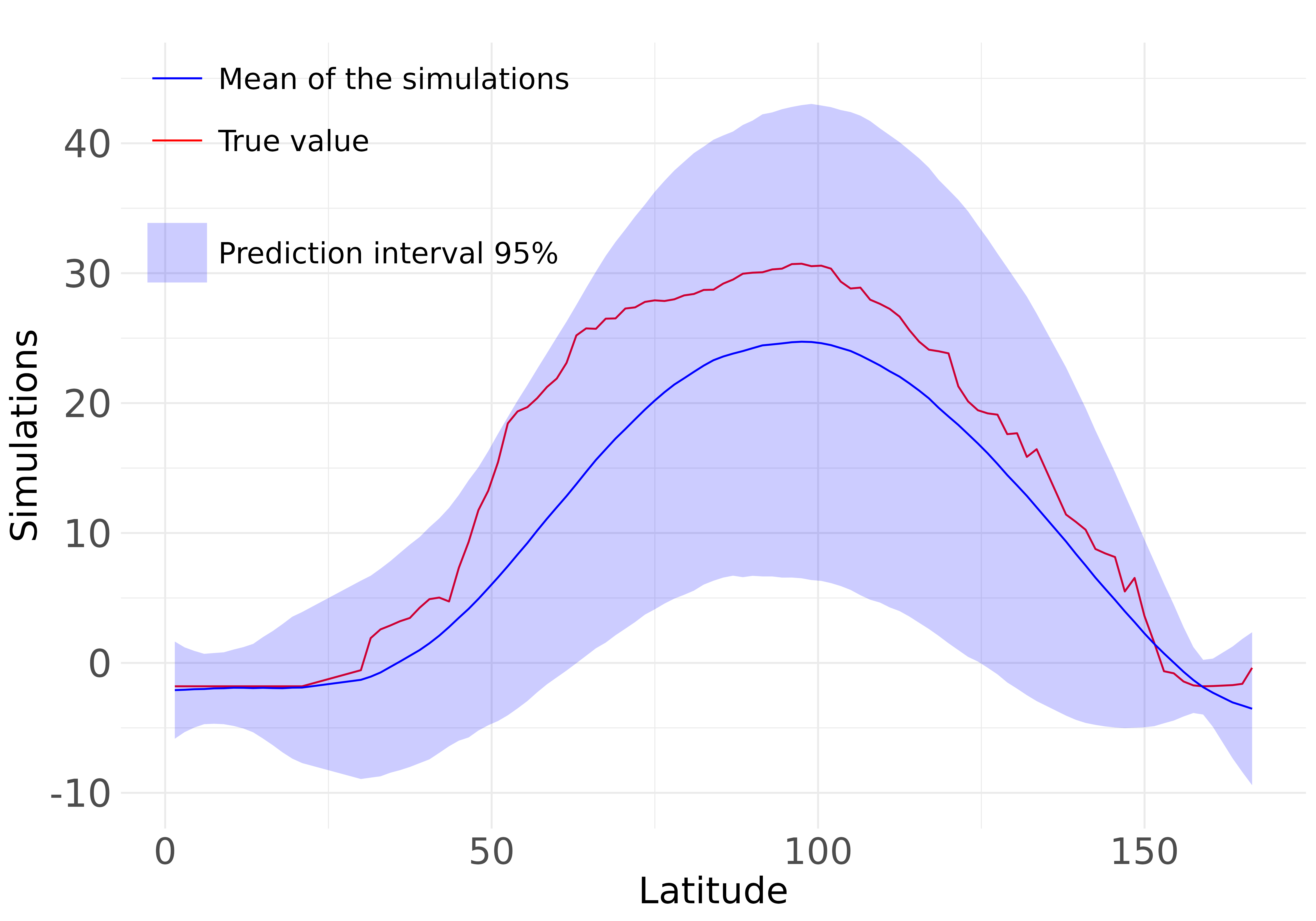}
\caption{Predictive distribution in the isotropic case.}
    \end{subfigure}
    \hfill
    \begin{subfigure}[t]{0.48\textwidth}
        \centering
        \includegraphics[width=\linewidth]{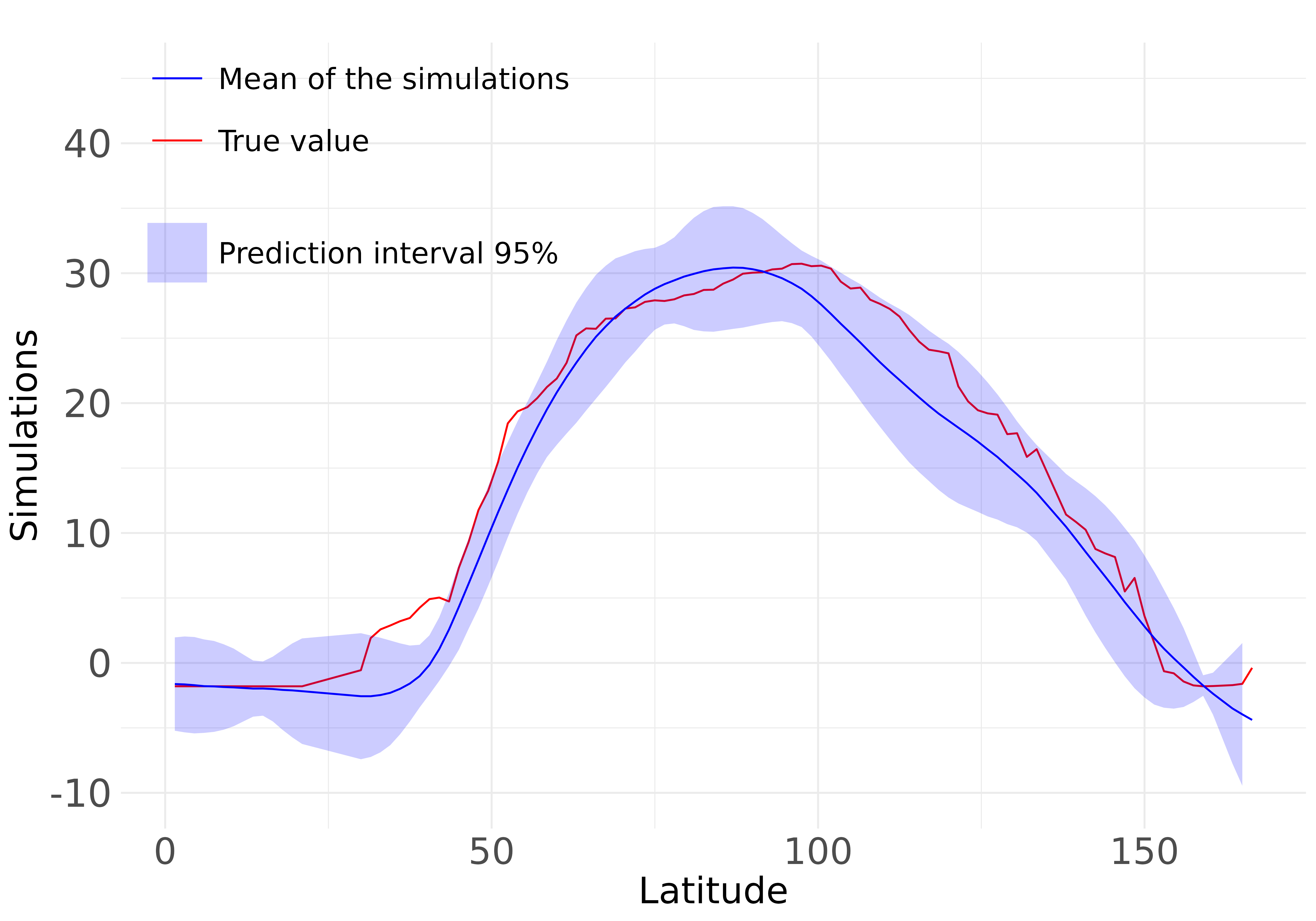}
        \caption{Predictive distribution in the anisotropic case.}
    \end{subfigure}
    
    \caption{Comparison between the uncertainty quantification in the isotropic and the anisotropic SST prediction on the Earth, along the longitude closest to $170^\circ$.}
    \label{uncert_earth}
\end{figure}

Note that the case of the sphere is particularly interesting, since the eigenfunctions of $-\Delta$ are the spherical harmonics~\cite{dai2013spherical}. As detailed in~\cite{sire2025spline}, the kernel
\[
K_1(\sbf_1, \sbf_2) 
= \sum_{k \in \mathbb{N}^\star} \frac{1}{\lambda_k^2} 
\phi_k(\sbf_1)\phi_k(\sbf_2)
\]
can therefore be computed explicitly by truncating the series expansion. In practice, the sum is truncated at $K = 40$, following the recommendations of~\cite{keller2019thin}. Classical kriging operations can then be performed using this kernel. The resulting predictions serve for comparison with the finite element approximation, in order to validate our numerical approach, and are provided in~\ref{app_harmonics}.

\subsection{Prediction on the surface of a cylinder}

Here, we consider a cylinder of radius $1$ and height $20$, since in many industrial applications the height of a cylinder is typically much larger than its radius (e.g., in the study of subsea pipelines~\cite{seth2021buckling}).  
We define $\mathcal{C}$ as the surface under investigation:
\[
\mathcal{C} = \left\{ (x,y,z) \in \mathbb{R}^2 \times [0,20] \mid x^2 + y^2 = 1 \right\}.
\]  
As a case study, we consider the three-dimensional Franke function~\cite{franke1979critical}, which is widely used for interpolating smooth functions~\cite{challacombe2015n,franke1996localization}. It is defined on $[0,1]^3$ by
\[
\begin{aligned}
g(x,y,z) 
&= 0.75 \exp\!\left(
- \frac{
(a_x(9x - 2))^2 + (a_y(9y - 2))^2 + (a_z(9z - 2))^2
}{4}
\right)
\\[0.6em]
&\quad + 0.75 \exp\!\left(
- \left(
\frac{(a_x(9x + 1))^2}{49} + \frac{(a_y(9y + 1))^2}{10} + \frac{(a_z(9z + 1))^2}{10}
\right)
\right)
\\[0.6em]
&\quad + 0.5 \exp\!\left(
- \frac{
(a_x(9x - 7))^2 + (a_y(9y - 3))^2 + (a_z(9z - 5))^2
}{4}
\right)
\\[0.6em]
&\quad - 0.2 \exp\!\left(
- \left(
(a_x(9x - 4))^2 + (a_y(9y - 7))^2 + (a_z(9z - 5))^2
\right)
\right),
\end{aligned}
\]
where $a_x$, $a_y$, and $a_z$ are coefficients that introduce anisotropy, set to $a_x = 0.4$, $a_y = 0.4$, and $a_z = 1$. The phenomenon on $\mathcal{C}$ is then modeled by the function $\tilde{g}$ defined as
\[
\forall (x,y,z)\in\mathcal{C}, \quad 
\tilde{g}(x,y,z) = g\left(\frac{x+1}{2}, \frac{y+1}{2}, \frac{z}{20}\right),
\]
so that $g$ is applied on $[0,1]^3$.

A triangulation mesh is constructed using a regular grid in the cylindrical coordinate chart $(\theta, z)$, with a $5^\circ$ resolution in $\theta$ and a resolution of $0.5$ in $z$. A total of $n=10$ observation points are selected as described in Section~\ref{earth}. The anisotropy parameters are estimated using the same approach as before, i.e. with the rotation angle and the ratio between the scaling parameters in the cylindrical coordinate chart. The resulting prediction mean is displayed in Figure~\ref{pred_cylinder}, while the prediction errors are shown in Figure~\ref{errors_cylinder}.

\begin{figure}
    \centering
    \begin{subfigure}{0.9\textwidth}  
        \centering
        \includegraphics[width=\textwidth]{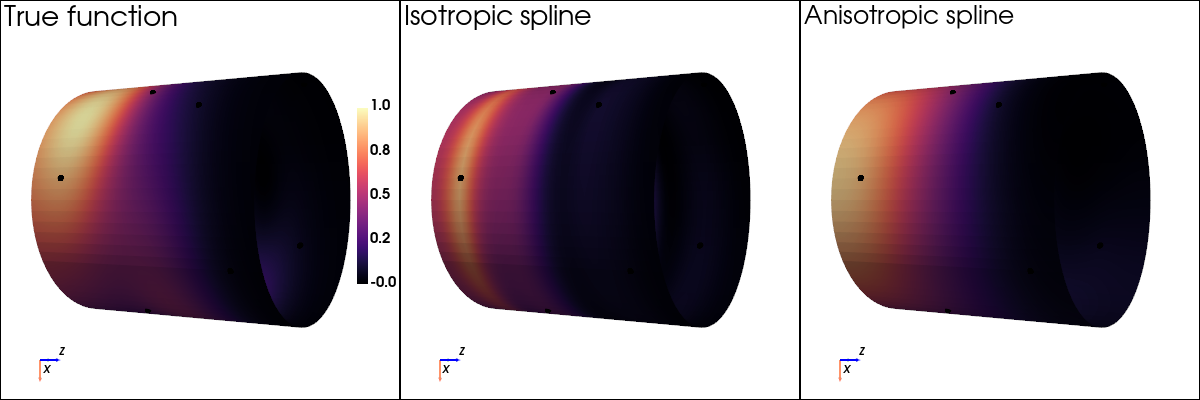} 
\caption{3D plots of the predictions compared to the true function. The cylinder's height is normalized for visualization purposes.}
    \end{subfigure}
    
    \vspace{1em} 
    
    \begin{subfigure}{\textwidth}
        \centering
        \includegraphics[width=\textwidth]{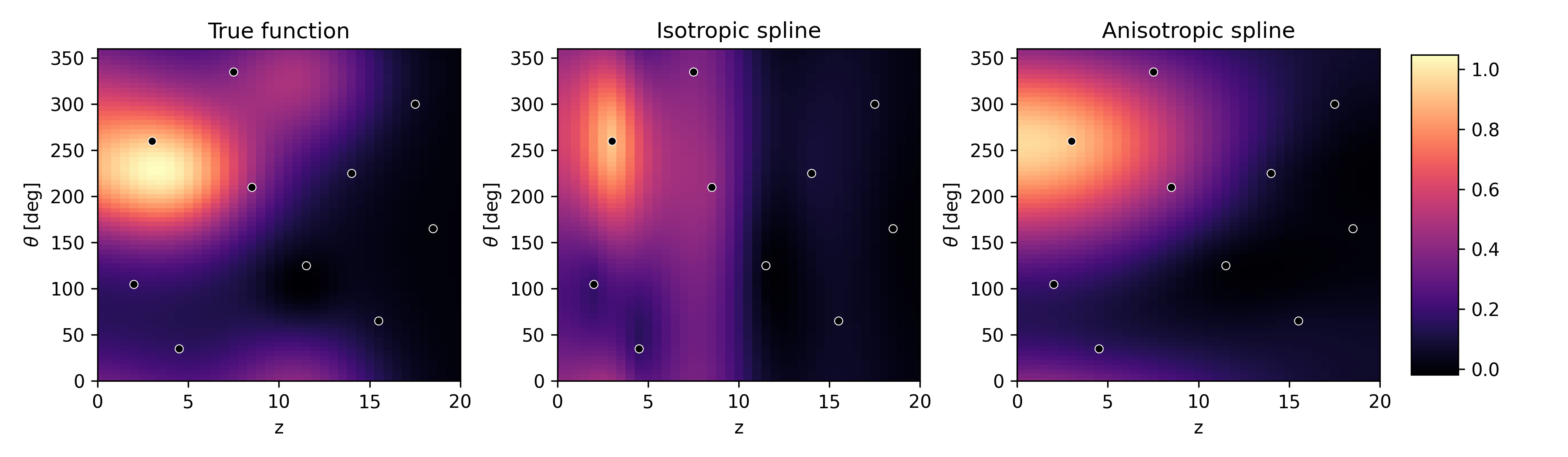} 
\caption{2D plots of the predictions compared to the true function. The $x$-axis represents the longitudinal direction, while the $y$-axis represents the polar angle in degrees.}
    \end{subfigure}
    
\caption{Predictions on the cylinder. Left: true function. Middle: isotropic spline. Right: anisotropic spline. $n=20$ observation points are shown as black dots.}
    \label{pred_cylinder}
\end{figure}

\begin{figure}
    \centering
    \begin{subfigure}[t]{0.48\textwidth}
        \centering
        \includegraphics[width=\linewidth]{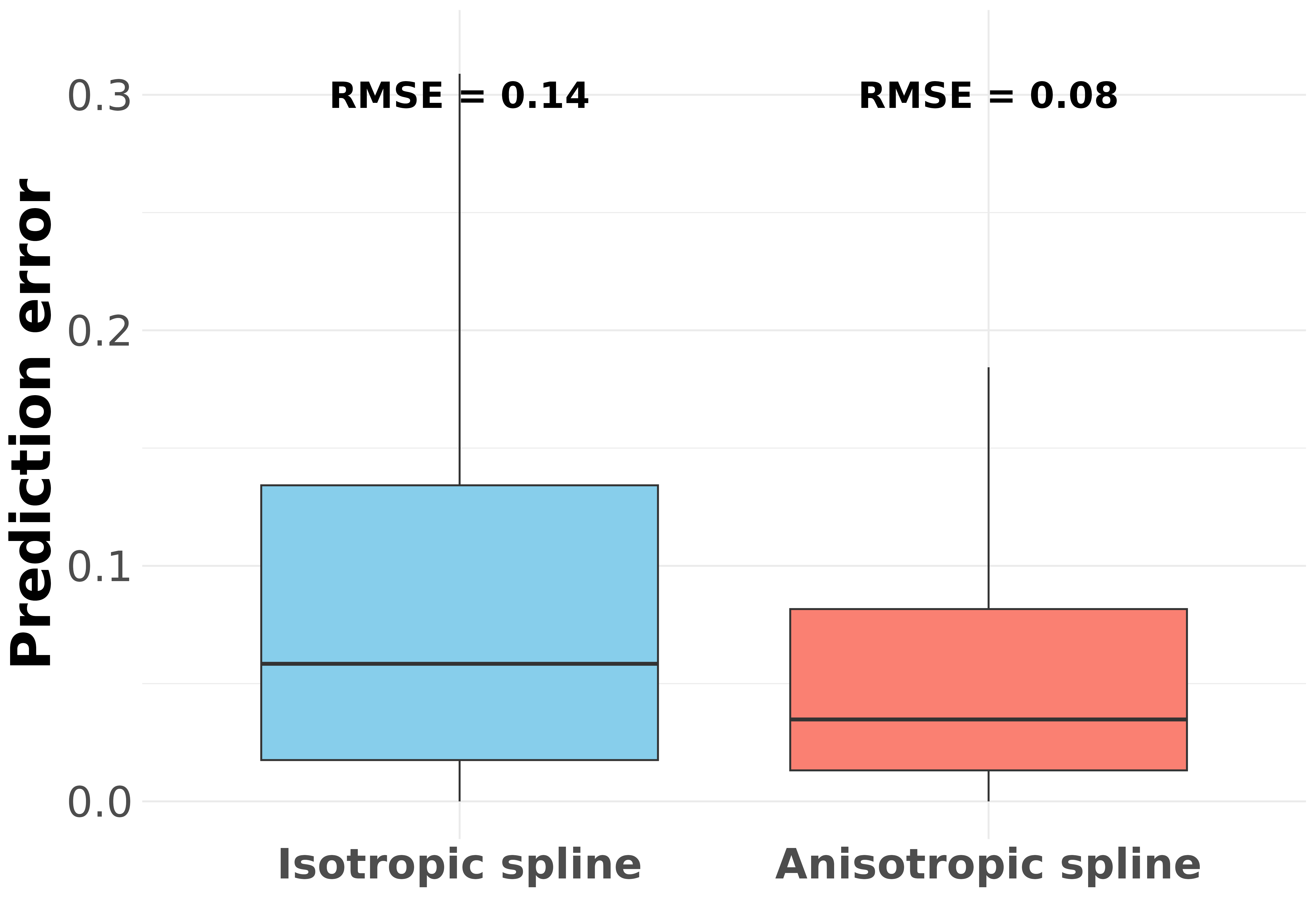}
\caption{Distribution of the prediction errors and associated RMSE in the isotropic and the anisotropic case.}
        \label{errors_cylinder}
    \end{subfigure}
    \hfill
    \begin{subfigure}[t]{0.48\textwidth}
        \centering
        \includegraphics[width=\linewidth]{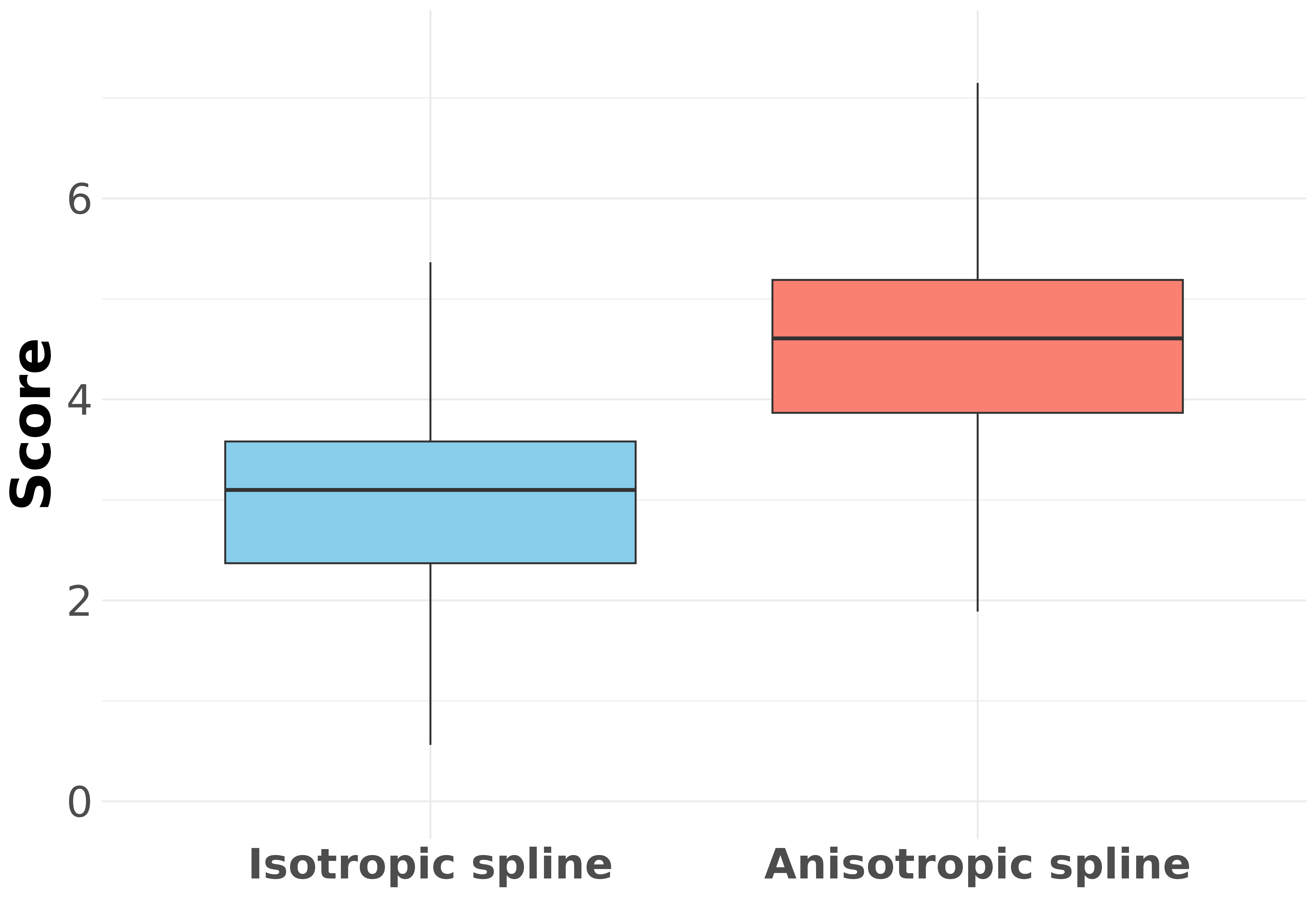}
        \caption{Distribution of the scores for the isotropic and the anisotropic prediction.}
        \label{scores_cylinder}
    \end{subfigure}
\caption{Comparison of the prediction errors and scores for the isotropic and anisotropic predictions on the cylinder.}
\end{figure}

Once again, estimating anisotropy substantially improves predictive accuracy. The uncertainty quantification appears consistent with the true values, as illustrated in Figure~\ref{uncert_cylinder}, which presents the $95\%$ prediction interval along the line $\theta = 270^\circ$. Finally, the distribution of the score defined in Equation~\ref{eq_score}, which balances predictive accuracy and uncertainty quantification, further confirms the improved performance obtained when incorporating anisotropy (see Figure~\ref{scores_cylinder}).

\begin{figure}
    \centering
    \begin{subfigure}[t]{0.48\textwidth}
        \centering
        \includegraphics[width=\linewidth]{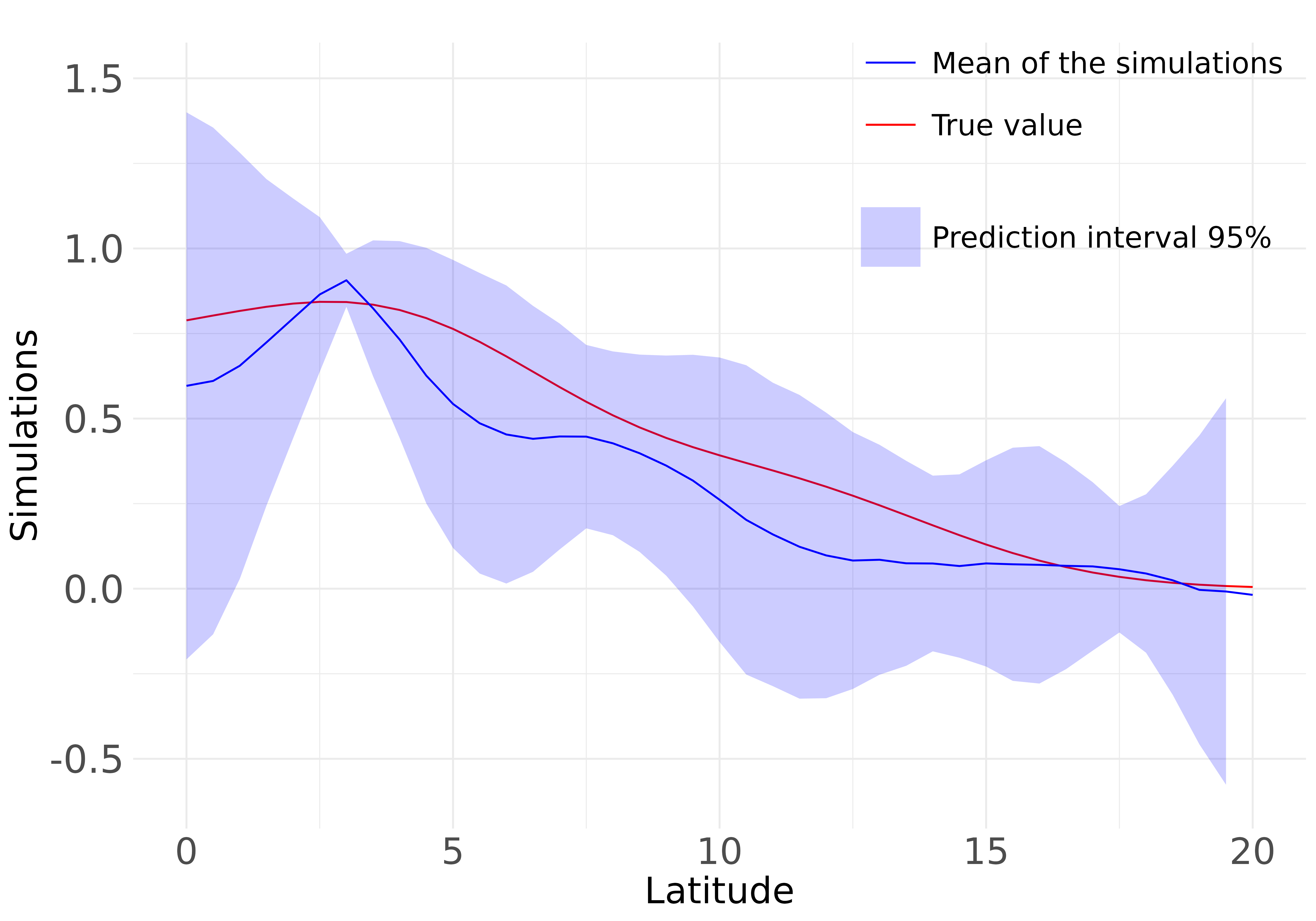}
\caption{Predictive distribution in the isotropic case.}
    \end{subfigure}
    \hfill
    \begin{subfigure}[t]{0.48\textwidth}
        \centering
        \includegraphics[width=\linewidth]{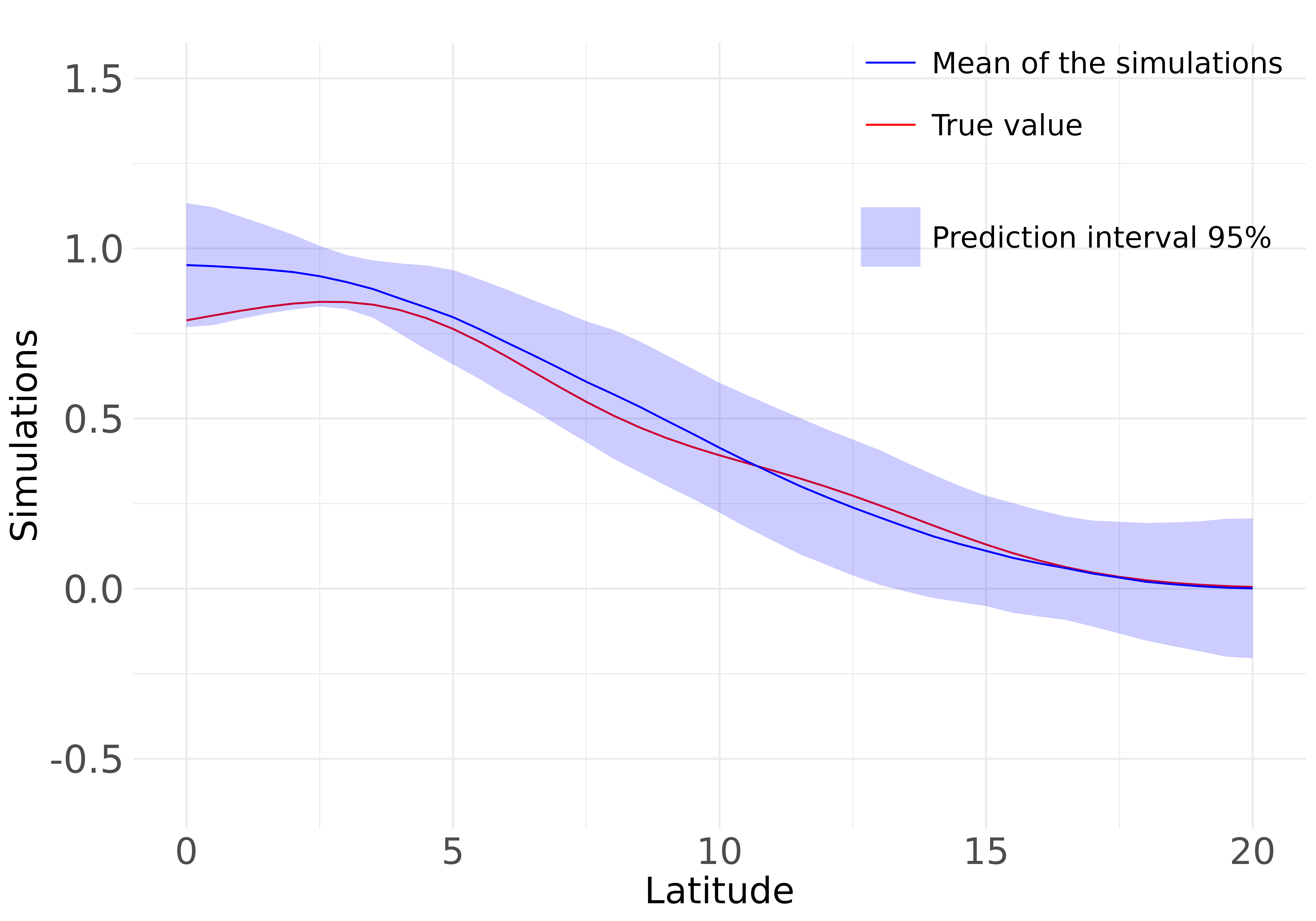}
        \caption{Predictive distribution in the anisotropic case.}
    \end{subfigure}
    \caption{Comparison between the uncertainty quantification in the isotropic and the anisotropic prediction on the cylinder, along the line $\theta = 270^\circ$.}
    \label{uncert_cylinder}
\end{figure}

\section{Summary and perspectives}\label{sec_conclu}

This article extends the work of~\cite{sire2025spline} by proposing a framework for spline interpolation on compact Riemannian manifolds, together with a quantification of the associated uncertainty. The approach relies on the equivalence between the classical solution to the interpolation problem and the universal kriging predictor. However, it is based on the spectrum of the Laplace-Beltrami operator $-\Delta$, which is generally unknown on compact Riemannian manifolds.

To overcome this difficulty, a finite-element approximation of the predictor is introduced, based on a triangulation of the manifold as proposed in~\cite{pereira_desassis_allard_2022}. This strategy avoids the need for the explicit spectrum of $-\Delta$ and enables the incorporation of anisotropies in the covariance kernel through a parametrization of the Riemannian metric associated with the manifold. However, it leads to an intrinsic Gaussian Markov random field (GMRF) with a singular covariance matrix, which complicates both computations and simulations.

Adapting the methodology of~\cite{bolin2021efficient}, we represent this intrinsic GMRF through linear conditioning of a non-intrinsic GMRF. The approach is illustrated on two test cases: the Earth, modeled as a spherical surface, and a cylinder. For both examples, the impact of anisotropy and the role of uncertainty quantification are emphasized. The predictive distribution is displayed over a selected subdomain to demonstrate that the predictive variance is consistent with the actual prediction error, even in situations where isotropic splines yield inaccurate point estimates. Finally, a scoring rule is evaluated to assess the overall predictive performance, taking into account both accuracy and uncertainty quantification. \\~\

Although this work offers a thorough investigation, several perspectives for future research remain. First, in this work we only considered constant anisotropy parameters. However, for certain phenomena, one may need to estimate a spatially varying anisotropy field. Such an extension would naturally lead to a very large number of parameters to infer through maximum log-likelihood estimation. Alternative strategies would therefore be worth exploring, for instance by modeling the anisotropy parameters themselves as a random field within a hierarchical framework. This would raise significant challenges, including the specification of a suitable prior statistical model for the anisotropy field and the implementation of a fully Bayesian approach. 

Another promising direction is to investigate the link between spline prediction and the Mat\'ern stochastic partial differential equation (SPDE), defined as follows~\cite{whittle1954stationary}:
\begin{equation}\label{eq_lindgren}
(\kappa^2 - \Delta)^{\alpha/2} Z = W,
\end{equation}
where $(\kappa^2 - \Delta)^{\alpha/2}$ is a pseudo-differential operator and $W$ denotes Gaussian white noise. As extensively studied in~\cite{bolin2025intrinsic} in $\mathbb{R}^d$, the solution to this SPDE is an intrinsic Gaussian field when $\kappa = 0$. On a compact manifold $\manif$, when $\alpha = 2$, this solution is strongly related to our spline predictor as it admits a covariance kernel $K_1$, and a then detailed investigation of this result would be of particular interest. More generally, this perspective opens the way to extending to Riemannian manifolds the work of~\cite{bolin2025intrinsic} on intrinsic Whittle--Mat\'ern fields, and to comparing our approach with theirs, which involves the computation of the Moore--Penrose pseudoinverse. It may then be interesting to study the behavior of our spline predictor, and more generally of intrinsic fields, in extrapolation scenarios, i.e. when predictions are made far from the observation points, and to compare it with results obtained using classical Matérn solutions.

Finally, in our work we used the score defined in Equation~\ref{eq_score} to assess the predictive performance at each point, while accounting for the associated uncertainty. This metric provides a convenient balance between accuracy and uncertainty quantification, and remains straightforward to compute when predictive simulations are available. However, it only relies on the marginal predictive variance at each location and does not incorporate the full covariance structure of the predictive distribution. More sophisticated scoring rules could therefore be considered, although their implementation becomes challenging when dealing with an intrinsic GMRF and a large number of prediction points (e.g., the entire triangulation mesh). In such situations, estimating the predictive covariance matrix empirically would require a very large number of independent simulations to obtain a sufficiently accurate estimator, while its analytical computation is delicate due to the manipulation of singular matrices.

\section*{Acknowledgements}
This work was fully supported by the Chaire Geolearning funded by Andra, BNP-Paribas, CCR and SCOR Foundation.

\section*{Declaration of Competing Interest}
The authors declare that they have no known competing financial interests or personal relationships that could have appeared to influence the work reported in this paper.

\section*{Data Availability}
The data and code used to generate the results presented in this article are provided in the Supplementary Material~\cite{charliesire_2026}.

\bibliographystyle{elsarticle-num-names} 
\bibliography{References}
\appendix
\section{Spectral Theorem}\label{spectral_theorem}

\begin{prop}[Spectrum of $-\Delta$]
The Spectral Theorem provides the following properties~\cite{Craioveanu2001}.
\begin{enumerate}[label=(\roman*)]
\item 
For the closed, the Dirichlet, and the Neumann eigenvalue problems, the spectrum of $-\Delta$ is an infinite, countable sequence of real eigenvalues
\[
0 \leq \lambda_0 \leq \lambda_1 \leq \dots \leq \lambda_k \leq \dots ,
\]
where each eigenvalue $\lambda_k$ appears with finite multiplicity $m_k$.

\item 
There exists an orthonormal basis $(\phi_k)_{k \in \mathbb{N}}$ of $\lm$ such that, for all $k \in \mathbb{N}$, $\phi_k \in \mathcal{C}^{\infty}(\manif)$ is an eigenfunction associated with $\lambda_k$. Consequently, for any $u \in \lm$,
\[
u = \sum_{k \in \mathbb{N}} u_k \phi_k,
\qquad
u_k = \langle u, \phi_k \rangle_{\lm}
     = \int_{\manif} u \, \phi_k \, d\mug .
\]

\item For the closed and the Neumann eigenvalue problems, the eigenvalue $\lambda_0 = 0$ corresponds to the constant functions \cite{urakawa1993geometry}, which implies that $m_0 = 1$ and
\[
\phi_0 = \frac{1}{\int_{\manif} d\mug}.
\]

\item For the Dirichlet eigenvalue problem, all eigenvalues are strictly positive.
\end{enumerate}
\end{prop}

In this study, we consider $\lambda_0 = 0$, as in the closed and Neumann eigenvalue problems, noting that the case $\lambda_0 > 0$ (Dirichlet problem) is computationally simpler for the operations that will be performed.

\section{Inversibility of $\Ktau$}\label{inversibility_ktau}

The inversibility of $\Ktau$ is guaranteed when $\noise >0.$ In the following, we focus on the case $\noise = 0,$ referring to the interpolating problem.

For $k \in \mathbb{N}^\star,$ let us denote 
$\wbf_k = \left(\phi_k(\sbf_1),\dots,\phi_k(\sbf_n)\right)^\top.$

\begin{align*}
\forall \xbf \in \mathbb{R}^n,~~\xbf^\top \Kzero \xbf &= \sum_{i=1}^n \sum_{j=1}^n  x_i x_j \sum_{k>0} \frac{1}{\lambda_k^2} \phi_k(\sbf_i) \phi_k(\sbf_j) \\
&= \sum_{k>0}  \frac{1}{\lambda_k^2} \sum_{i=1}^n \sum_{j=1}^n  x_i x_j \phi_k(\sbf_i) \phi_k(\sbf_j) \\
&= \sum_{k>0}  \frac{1}{\lambda_k^2} \left\lVert \sum_{i=1}^n  x_i \phi_k(\sbf_i) \right \rVert ^2
\end{align*} 
Then, \begin{align*}
\xbf^\top \Kzero \xbf = 0 &\Leftrightarrow  \forall k >0, \sum_{i=1}^n  x_i \phi_k(\sbf_i)=0 \\
&\Leftrightarrow \forall k >0, \langle \xbf, \wbf_k\rangle_{\mathbb{R}^n} = 0
\end{align*}
It comes

\begin{align*}
\Kzero\text{ is positive definite} &\Leftrightarrow \operatorname{span}\{\wbf_k\}_{k > 0} = \mathbb{R}^n\\
&\Leftrightarrow \forall \xbf \in \mathbb{R}^{n}, \exists \left(a_k\right)_{k>0}, \xbf = \sum_{k>0}a_k\wbf_k\\
&\Leftrightarrow \forall \xbf \in \mathbb{R}^{n}, \exists \left(a_k\right)_{k>0}, \forall\, 1\leq i \leq n,~ x_i = \sum_{k>0}a_k\phi_k(\sbf_i),\\
&\Leftrightarrow \forall \xbf \in \mathbb{R}^{n}, \exists \left(a_k\right)_{k>0}, \text{ the function } u = \sum_{k>0} a_k \phi_k \text{ interpolates the vector }\xbf
\end{align*}

As $\left(\phi_k\right)_{k>0}$ is a basis of $\lm,$ $\Kzero$ is invertible if and only if the an interpolating function exists in $\lm$ for any vector of observations $\xbf.$

Note that this condition cannot be satisfied when an observation point lies on the boundary $\partial M$ under Dirichlet boundary conditions.

\section{Proof of Proposition~\ref{prop_conv_bayes}}\label{proof_bayes}

Our random vector of observations $\Yobs = (Y_i)_{i=1}^{n}$ verifies

$$Y_i = Z(\sbf_i)+\sigma\tau E_i, ~~1\leq i \leq n,$$
\noindent with 
\begin{itemize}
\item $(Z\mid A = a)$ a GRF with mean $a\phi_0$ and covariance kernel $\sigma^2 K_1$, for all $a\in \mathbb{R}$
\item $A\sim\mathcal{N}(\mu, \alpha)$
\item $\Ebf = (E_i)_{i=1}^{n}$ a centered Gaussian vector with $\cov(\Ebf) = \Ibf_n$, independent of $A$ and $Z$.
\end{itemize}

The objective is to compute the mean and covariance kernel of $\left(Z(\sbf)\mid \Yobs = \ybf\right)$ when $\alpha$ tends to $+\infty.$

\subsection{Distribution of $Z(\sbf)$} \label{app_distrib_z}

Here, we will show that $Z(\sbf)$ is a GRF with
\begin{itemize}
\item mean $\mu\phi_0$ 
\item covariance kernel $\alpha \phi_0(\sbf)\phi_0(\sbf')+\Sigma^2 K_1(\sbf,\sbf')$
\end{itemize}

Note that this step is not mandatory, but is included here since a very similar development can be used to show that $\Zbf \sim \mathcal{N}(\veczerm,\sigalp)$ in Section~\ref{general_idea}.

Let us consider $\xbf_1, \dots, \xbf_p$, $p$ points of $\manif,$ and study the distribution of $$\Ztest = Z(\xbf_1),\dots, Z(\xbf_p).$$ We denote $\phitest = \left(\phi_0(\xbf_i)\right)_{i=1}^{p},$ and $\Ktest = \left[K_1(\xbf_i,\xbf_j)\right]_{1 \leq i,j \leq p}$. Then density of $\Ztest$ is
\begin{align*}
p_{\Ztest}(\zbf) &= \int_{\mathbb{R}^p}p_{\Ztest\mid A = a}(\zbf)p_A(a) da \\
&=\int_{\mathbb{R}}
\frac{1}{(2\pi)^{(p+1)/2}\lvert(\sigma^2\Ktest)\rvert^{1/2}}
\exp\!\left(
-\frac{1}{2}
(\zbf-a\phitest)^\top (\sigma^2\Ktest)^{-1}(\zbf-a\phitest)
\right)
\frac{1}{\sqrt{\alpha}}
\exp\!\left(
-\frac{(a-\mu)^2}{2\alpha}
\right)da\\
&=\frac{1}{\sqrt{\alpha}(2\pi)^{(p+1)/2}\lvert(\sigma^2\Ktest)\rvert^{1/2}}\exp\left(-\frac{1}{2}\left(\zbf^\top(\sigma^2\Ktest)^{-1}\zbf+\frac{\mu^2}{\alpha}\right) \right)\int_{\mathbb{R}}\exp\!\left(
-\frac{1}{2}\gamma a^2 + \zeta a
\right)da \\
\text{with } &\gamma = \left((\phitest)^\top(\Ktest)^{-1}\phitest+\frac{1}{\alpha}\right) \text{ and } \zeta = \zbf^\top(\sigma^2\Ktest)^{-1}\phitest+\frac{\mu}{\alpha}.
\end{align*}

Then,
\begin{align*}
p_{\Ztest}(\zbf)&= \frac{1}{(2\pi)^{(p+1)/2}\lvert(\sigma^2\Ktest)\rvert^{1/2}}\exp\left(-\frac{1}{2}\zbf^\top(\sigma^2\Ktest)^{-1}\zbf+\frac{\mu^2}{\alpha}\right)\sqrt{\frac{2\pi}{\gamma}}\,
\exp\!\left(\frac{\zeta^2}{2\gamma}\right) \\
&\propto \exp\left(-\frac{1}{2}\zbf^\top\left[(\sigma^2\Ktest)^{-1} - \frac{1}{\gamma}(\sigma^2\Ktest)^{-1}\phitest(\phitest)^\top(\Ktest)^{-1}\right]\zbf\right)\exp\left(\frac{1}{\gamma}\zbf^\top(\sigma^2\Ktest)^{-1}\phitest\frac{\mu}{\alpha}\right)\\
\end{align*}

Let us denote $\Kalp = \sigma^2\Ktest + \alpha\phitest(\phitest)^\top.$ From Woodbury formula, $$\Kalp^{-1} = (\sigma^2\Ktest)^{-1} - \frac{1}{\gamma}(\sigma^2\Ktest)^{-1}\phitest(\phitest)^\top(\sigma^2\Ktest)^{-1}.$$

And as we can remark that $\frac{1}{\gamma\alpha} = 1 - \frac{(\phitest)^\top(\sigma^2\Ktest)^{-1} \phitest}{\gamma},$ we have

\begin{align*}
\frac{1}{\gamma}\zbf^\top(\sigma^2\Ktest)^{-1}\phitest\frac{\mu}{\alpha} &= \zbf^\top\left((\sigma^2\Ktest)^{-1} - \frac{1}{\gamma}(\sigma^2\Ktest)^{-1}\phitest\left(\phitest\right)^\top(\sigma^2\Ktest)^{-1}\right)\mu\phitest\\
&= \zbf^\top\Kalp^{-1}\mu\phitest.
\end{align*}

Finally,

$$p_{\Ztest}(\zbf)\propto \exp\left(-\frac{1}{2}\zbf^\top\Kalp^{-1}\zbf + \zbf^\top\Kalp^{-1}\mu\phitest\right).$$

We recognize the density of a Gaussian distribution, with mean vector $\mu\phitest$ and covariance matrix $\Kbf_{\alpha}.$ That shows that $Z$ is a GRF with function mean $\mu\phi_0$ and covariance kernel $\alpha\phi_0(s)\phi_0(s')+\sigma^2 K_1(s,s').$

\subsection{Distribution of $\left(Z(\sbf) \mid \Yobs = \ybf\right)$} \label{app_distrib_condi}

From 
\begin{equation*}
\left\{
\begin{aligned}
&\Yobs = \left(Z(\sbf_i) +\sigma\noise E_i\right)_{i=1}^{n}\\
&Z \text{ is a GRF of mean } \mu\phi_0 \text{ and covariance kernel } \sigma^2K_1(\sbf,\sbf')+\alpha\phi_0(\sbf)\phi_0(\sbf'),
\end{aligned}
\right.
\end{equation*}
we use the conditional Gaussian formulas and obtain that $\left(Z(\sbf)\mid\Yobs = \ybf\right)$ is a GRF with
\begin{itemize}
\item mean $\moy(\sbf) = \mu\phi_0(\sbf) + \left[\sigma^2\kbf_1(\sbf) + \alpha\phi_0(\sbf)\phibfn\right]^\top\left[\sigma^2\Ktau +\alpha\phibfn(\phibfn)^\top\right]^{-1}\left[\ybf - \mu\phibfn\right]$
\item covariance kernel \begin{align*}
\covk&(\sbf,\sbf') = \sigma^2 K_1(\sbf,\sbf')+\alpha\phi_0(\sbf)\phi_0(\sbf')\\
&-  \left[\sigma^2\kbf_1(\sbf) + \alpha\phi_0(\sbf)\phibfn\right]^\top\left[\sigma^2\Ktau +\alpha\phibfn(\phibfn)^\top\right]^{-1}\left[\sigma^2\kbf_1(\sbf') + \alpha\phi_0(\sbf')\phibfn\right]\end{align*}
\end{itemize}

\subsection{Limit when $\alpha \to \infty$}

Now, let us compute the limits $\underset{\alpha \to \infty}{\lim} \moy(\sbf)$ and $\underset{\alpha \to \infty}{\lim} \covk(\sbf,\sbf').$

First, from Woodbury formula, 
\begin{align*}
\left[\sigma^2\Ktau +\alpha\phibf(\phibfn)^\top\right]^{-1} = \sigma^{-2}\Ktau^{-1}\left[1 - \sigma^{-2}\alpha\phibfn \frac{(\phibfn)^\top\Ktau^{-1}}{1 + \sigma^{-2}\alpha (\phibfn)^\top\Ktau^{-1}\phibfn}\right]
\end{align*}
Then, 
\begin{align*}
&\alpha\phi_0(\sbf)(\phibfn)^\top\left[\sigma^{-2}\Ktau +\alpha\phibf(\phibfn)^\top\right]^{-1} \\
&= \sigma^{-2}\alpha\phi_0(\sbf)(\phibfn)^\top\Ktau^{-1}- \sigma^{-2}\alpha\phi_0(\sbf)\frac{\sigma^{-2}\alpha(\phibfn)^\top\Ktau^{-1}\phibfn}{1 + \sigma^{-2}\alpha (\phibfn)^\top\Ktau^{-1}\phibfn}(\phibfn)^\top\Ktau^{-1}\\
&=\sigma^{-2}\alpha\phi_0(\sbf)(\phibfn)^\top\Ktau^{-1} - \sigma^{-2}\alpha\phi_0(\sbf)\left(1 - \frac{1}{1 + \sigma^{-2}\alpha (\phibfn)^\top\Ktau^{-1}\phibfn}\right)(\phibfn)^\top\Ktau^{-1}\\
&=\frac{\sigma^{-2}\phi_0(\sbf)(\phibfn)^\top\Ktau^{-1}}{\frac{1}{\alpha} + \sigma^{-2}(\phibfn)^\top\Ktau^{-1}\phibfn}
\end{align*}
and \begin{align*}
\sigma^2\kbf_1(\sbf)^\top\left[\sigma^2\Ktau +\alpha\phibf(\phibfn)^\top\right]^{-1} = \kbf_1(\sbf)^\top\Ktau^{-1} - \frac{\sigma^{-2}\kbf_1(\sbf)^\top\Ktau^{-1}\phibfn(\phibfn)^\top\Ktau^{-1}
}{\frac{1}{\alpha} + \sigma^{-2}(\phibfn)^\top\Ktau^{-1}\phibfn}
\end{align*}

\paragraph*{Convergence of $\moy(\sbf)$} 

Then, with $\frac{1}{\alpha}\to 0,$ it gives 

\begin{align*}
\underset{\alpha \to \infty}{\lim} \moy(\sbf) &= \mu\phi_0(\sbf) + \left[ \kbf_1(\sbf)^\top\Ktau^{-1}
+ \left[\phi_0(\sbf)-\kbf_1(\sbf)^\top\Ktau^{-1}\phibfn\right]\frac{(\phibfn)^\top\Ktau^{-1}
}{(\phibfn)^\top\Ktau^{-1}\phibfn}\right]\left[\ybf - \mu\phibfn\right] \\
&= \kbf_1(\sbf)^\top\Ktau^{-1}\ybf + \amnoise \left[\phi_0(\sbf)-\kbf_1(\sbf)^\top\Ktau^{-1}\phibfn\right]\\
&\;\;\;\;+ \mu\left[\phi_0(\sbf) - \kbf_1(\sbf)^\top\Ktau^{-1}\phibf - \left[\phi_0(\sbf)-\kbf_1(\sbf)^\top\Ktau^{-1}\phibfn\right]\right] \\
&=u^\star_\noise(\sbf).
\end{align*}

\paragraph*{Convergence of $\covk(\sbf,\sbf')$}

\begin{align*}
 &\left[\sigma^2\kbf_1(\sbf) + \alpha\phi_0(\sbf)\phibfn\right]^\top\left[\sigma^2\Ktau +\alpha\phibfn(\phibfn)^\top\right]^{-1}\left[\sigma^2\kbf_1(\sbf') + \alpha\phi_0(\sbf')\phibfn\right] \\
 &= \kbf_1(\sbf)^\top\Ktau^{-1}\left[\sigma^2\kbf_1(\sbf') + \alpha\phi_0(\sbf')\phibfn\right] \\
 &+\left[\phi_0(\sbf)-\kbf_1(\sbf)^\top\Ktau^{-1}\phibfn\right]\frac{\sigma^{-2}(\phibfn)^\top\Ktau^{-1}
}{\frac{1}{\alpha}+\sigma^{-2}(\phibfn)^\top\Ktau^{-1}\phibfn}\left[\sigma^2\kbf_1(\sbf') + \alpha\phi_0(\sbf')\phibfn\right]\\
&= \sigma^2\kbf_1(\sbf)^\top\Ktau^{-1}\kbf_1(\sbf') +\left[\phi_0(\sbf)-\kbf_1(\sbf)^\top\Ktau^{-1}\phibfn\right]\frac{(\phibfn)^\top\Ktau^{-1}
}{\frac{1}{\alpha}+\sigma^{-2}(\phibfn)^\top\Ktau^{-1}\phibfn}\kbf_1(\sbf')\\
&+ \alpha \kbf_1(\sbf)^\top\Ktau^{-1}\phi_0(\sbf')\phibfn + \alpha\left[\phi_0(\sbf)-\kbf_1(\sbf)^\top\Ktau^{-1}\phibfn\right]\frac{\sigma^{-2}\alpha(\phibfn)^\top\Ktau^{-1}\phibfn
}{1+\sigma^{-2}\alpha(\phibfn)^\top\Ktau^{-1}\phibfn}\phi_0(\sbf')\\
&= \sigma^2\kbf_1(\sbf)^\top\Ktau^{-1}\kbf_1(\sbf') +\left[\phi_0(\sbf)-\kbf_1(\sbf)^\top\Ktau^{-1}\phibfn\right]\frac{(\phibfn)^\top\Ktau^{-1}
}{\frac{1}{\alpha}+\sigma^{-2}(\phibfn)^\top\Ktau^{-1}\phibfn}\kbf_1(\sbf')\\
&+ \alpha \kbf_1(\sbf)^\top\Ktau^{-1}\phi_0(\sbf')\phibfn + \alpha\left[\phi_0(\sbf)-\kbf_1(\sbf)^\top\Ktau^{-1}\phibfn\right]\left[1-\frac{1}{1+\sigma^{-2}\alpha(\phibfn)^\top\Ktau^{-1}\phibfn}\right]\phi_0(\sbf') \\
&= \sigma^2\kbf_1(\sbf)^\top\Ktau^{-1}\kbf_1(\sbf') +\left[\phi_0(\sbf)-\kbf_1(\sbf)^\top\Ktau^{-1}\phibfn\right]\frac{(\phibfn)^\top\Ktau^{-1}
}{\frac{1}{\alpha}+\sigma^{-2}(\phibfn)^\top\Ktau^{-1}\phibfn}\kbf_1(\sbf')\\
&+ \alpha \phi_0(\sbf)\phi_0(\sbf') - \left[\phi_0(\sbf)-\kbf_1(\sbf)^\top\Ktau^{-1}\phibfn\right]\frac{\phi_0(\sbf')}{\frac{1}{\alpha}+\sigma^{-2}(\phibfn)^\top\Ktau^{-1}\phibfn}
\end{align*} 
 And then it comes
\begin{align*}
&\underset{\alpha \to \infty}{\lim} \covk(\sbf,\sbf') = \sigma^2 K_1(\sbf,\sbf')  -  \sigma^2\kbf_1(\sbf)^\top\Ktau^{-1}\kbf_1(\sbf')  \\
&+ \sigma^2\left[\phi_0(\sbf)-\kbf_1(\sbf)^\top\Ktau^{-1}\phibfn\right]\left[(\phibfn)^\top\Ktau^{-1}\phibfn\right]^{-1}\left[\phi_0(\sbf')-(\phibfn)^\top\Ktau^{-1}\kbf_1(\sbf')\right].
\end{align*} 

\paragraph*{Convergence to a GRF}

For any finite set of $p$ test points $\{\xbf_1, \dots, \xbf_p\} \subset \manif$, let us denote $\Ztest = (Z(\xbf_1), \dots, Z(\xbf_p))^\top$ the corresponding random vector. Since $(Z \mid \Yobs = \ybf)$ is a Gaussian Random Field for any $\alpha >0$, the conditional vector $(\Ztest \mid \Yobs = \ybf)$ is Gaussian with a mean vector and a covariance matrix that converge, as $\alpha \to \infty$, to $\mtest$ and $\ktest$ respectively. It follows from Lévy's continuity theorem that
\begin{equation}
(\Ztest \mid \Yobs = \ybf) \xrightarrow[\alpha \to \infty]{\mathcal{L}} \mathcal{N}(\mtest, \ktest),
\end{equation}
which ensures that all finite-dimensional distributions of the limiting process are Gaussian, thereby proving that the limit of $(Z \mid \Yobs = \ybf)$ remains a GRF.

\section{Definition of $-\Delta_m$}\label{galerk}

This section defines the Galerkin approximation $-\Delta_m$ of the Laplace-Beltrami opperation $-\Delta.$ Let us denote $V_m$ the linear span of $\psi_1,\dots,\psi_m$ the piecewise linear functions associated with the triangulation $\tri.$ $-\Delta_m$ is the endomorphism mapping any $f \in V_m$ to $-\Delta_m f,$ with 

$$\forall u \in V_m, \langle -\Delta_m f, u \rangle_{\lm} = \langle - \Delta f, u \rangle_{\lm}.$$

\section{Computation of $\phibf$}\label{formula_phibf}

By definition, the function $\phi_0^{(m)}$ has unit norm in $\lm$, that is,
\begin{equation*}
\left\langle \sum_{k=1}^{m} [\phibf]_k \psi_k,\ \sum_{k=1}^{m} [\phibf]_k \psi_k \right\rangle_{\lm} = 1,
\end{equation*}
which implies that 
\begin{equation*}
\phibf^\top \Mbf \phibf = \left( \sqrt{\Mbf}^\top \phibf \right)^\top \sqrt{\Mbf}^\top \phibf = 1.
\end{equation*}

Since $\phibf$ is constant, it follows that
\begin{equation}
\phibf = \frac{\mathbf{1}_m}{\lVert \sqrt{\Mbf}^\top \mathbf{1}_m \rVert_2}. 
\end{equation}

\section{Inversibility of $\Kbis$}\label{inversibility_kbis}

\subsection{First scenario}\label{first_scenario_appendix}

As a reminder, $\sigbf$ is positive semidefinite of rank $ m - 1 $. Now let us identify a vector in $\text{Ker}(\sigbf).$ By definition of the basis function $\left(\psi_j\right)_{j=1}^{m},$ we have $\forall \sbf \in \manif, \sum_{j=1}^{m} \psi_j(\sbf) = 1,$ and then 
$$\forall \sbf \in \manif, \sum_{j=1}^{m} \nabla\psi_j(\sbf) = 0.$$

As $\phibf$ is constant, $\Sbf \left(\sqrt{\Mbf}\right)^{\top}\phibf = \left(\sqrt{\Mbf}\right)^{-1}\Fbf\left(\sqrt{\Mbf}\right)^{-\top}\left(\sqrt{\Mbf}\right)^{\top}\phibf = \left(\sqrt{\Mbf}\right)^{-1}\Fbf\phibf = 0.$ Then, $\left(\sqrt{\Mbf}\right)^{\top}\phibf \in \text{Ker}(\Sbf)\subset \text{Ker}(f(\Sbf)).$

It comes, $$\sigbf\Mbf\phibf = \left(\sqrt{\Mbf}\right)^{-\top} f(\Sbf) \left(\sqrt{\Mbf}\right)^{-1}\sqrt{\Mbf}\left(\sqrt{\Mbf}\right)^\top\phibf = 0.$$
And finally $\ker(\sigbf) = \text{span}(\Mbf\phibf).$

As a reminder, $ I = \{i_1, \dots, i_n\} $ denotes the set of indices such that $ s_k = c_{i_k} $ for $ 1 \leq k \leq n $. Then, $\forall \ybf \in\mathbb{R}^n \setminus\{0\},$ 
$$
\left[\Abf_n^\top\ybf\right]_i =
\begin{cases}
\ybf_k & \text{if } i = i_k \text{ for some } 1 \leq k \leq n, \\
0 & \text{otherwise},
\end{cases}
$$ and $\Abf_n^\top\ybf \notin \ker(\sigbf).$ It shows that $\Abf_n\sigbf\Abf_n^\top$ is positive definite.

\subsection{Second scenario}

For the \emph{second scenario}, the covariance matrix $ \Abf_n \sigbf \Abf_n^\top + \noise^2 \Ibf_n $ is invertible if and only if $ -\noise^2 $ is not an eigenvalue of $ \Abf_n \sigbf \Abf_n^\top $. As $\sigbf$ is positive semidefinite, $\Abf_n \sigbf \Abf_n^\top$ is positive semidefinite, and $-\noise^2$ is not an eigenvalue of $ \Abf_n \sigbf \Abf_n^\top .$

\section{Simulations with known mean}\label{other_simu}

Here, we provide a strategy to generate simulations of $\left(Z^{(m)} \mid \Yobs^{(m)} = \ybf, A = \amnoise\right).$ From Proposition~\ref{prop_sigma}, we have $\left(Z^{(m)}\mid A = \amnoise\right)$ is a Gaussian Process with mean $\amnoise\phi_0(\sbf)$ and covariance kernel $\sigma^2 \Abf(\sbf)^\top\sigbf \Abf(\sbf').$ Then, using conditionnal Gaussian process formulas, $$\left(Z^{(m)} \mid \Yobs^{(m)} = \ybf, A = \amnoise\right)$$ is a GRF with 
\begin{itemize}
\item Mean $\amnoise \Abf(\sbf)^\top\phibf +\Abf(\sbf)^\top\sigbf \Abf_n^\top \left(\Kbis\right)^{-1}(\ybf - \amnoise\Abf_n\phibf)$
\item Covariance kernel $\sigma^2\left[\Abf(\sbf)^\top\sigbf \Abf(\sbf') - \Abf(\sbf)^\top\sigbf \Abf_n^\top \left(\Kbis \right)^{-1}\Abf_n\sigbf \Abf(s')\right].$
\end{itemize}

This GRF has the same mean as $\left(Z^{(m)} \mid \Yobs^{(m)} = \ybf\right)$ in the limit as $\alpha \to \infty$,, while the covariance kernels match except for the term $c^{(m)}(\sbf,\sbf')$, which accounted for the uncertainty in the mean parameter $A$ (see Proposition~\ref{prop_conv_fe}). Therefore, it remains relevant to generate simulations of $\left(Z^{(m)} \mid \Yobs^{(m)} = \ybf, A = \amnoise\right)$: the predictive mean coincides with the target value, and the associated uncertainty corresponds to the simple kriging framework. Let us introduce the following proposition to obtain the target simulations.

\begin{prop}[Simulations using the simple kriging framework]\label{prop_sk}
Let $\Zbf_0 \sim \mathcal{N}(\veczerm, \sigma^2\sigbf)$ and $\tilde{\Zbf}_0 = \Zbf_0 + \phibf\frac{(\Mbf\phibf)^\top\mpostalp(\ybf)}{(\Mbf\phibf)^\top\mpostalp(\Abf_n\phibf)}.$ Then, the vector
$$\Zsk = \tilde{\Zbf}_0 + \left[\Ibf_m - \phibf(\Mbf\phibf)^\top+h\left[\mpostalp(\Abf_n\phibf)\right](\Mbf\phibf)^\top\right]\mpostalp\left(\ybf - \Abf_n\tilde{\Zbf}_0 - \sigma\noise\Ebf\right)$$ has same distribution as $\left(\Zbf \mid \Yobs^{(m)} = \ybf, A = \amnoise\right).$
\end{prop}

\begin{proof}

Let us denote

\begin{equation*}
\left\{
\begin{aligned} 
\msk(\ybf) &= \esp\left(\Zbf \mid \Yobs^{(m)} = \ybf, A = \amnoise\right) = \mpost(\ybf)\\
\covsk &= \cov\left(\Zbf \mid \Yobs^{(m)} = \ybf, A = \amnoise\right) = \sigma^2\left[\sigbf - \sigbf \Abf_n^\top \left(\Kbis \right)^{-1}\Abf_n\sigbf\right].
\end{aligned}
\right.
\end{equation*}

As detailed in~\cite{sire2025spline},
\begin{equation*}
\left\{
\begin{aligned}
(\Mbf\phibf)^\top \mpostalp(\xbf) &= \frac{1}{\frac{1}{\alpha}+ \sigma^{-2}(\Abf_n\phibf)^\top\left(\Kbis\right)^{-1}\Abf_n\phibf}\sigma^{-2}(\Abf\phibf)^\top\left(\Kbis\right)^{-1}\xbf,~~\forall x \in \mathbb{R}^{m}\\
h\left[\mpostalp(\Abf_n\phibf)\right] &= \sigbf\Abf_n^\top\left(\Kbis\right)^{-1}\Abf_n\phibf.
\end{aligned}
\right.
\end{equation*}

Then, \begin{equation}\label{eq1proof}
\tilde{\Zbf}_0 = \Zbf_0 + \amnoise\phi_0.
\end{equation}

From Proposition~\ref{compute_mean}, 
\begin{align*}
\Zsk &= \tilde{\Zbf}_0 + \mpost\left(\ybf - \Abf_n\tilde{\Zbf}_0 - \sigma\noise\Ebf\right)+ \frac{\hh\!\left[\mpostalp\!\left(\Abf_n \phibf\right)\right]-\phibf}{\left(\Mbf\phibf\right)^\top \mpostalp\!\left(\Abf_n \phibf\right)}(\Mbf\phibf)^\top\mpostalp\left(\ybf - \Abf_n\tilde{\Zbf}_0 - \sigma\noise\Ebf\right)\\ 
= &\tilde{\Zbf}_0 + \mpost\left(\ybf - \Abf_n\tilde{\Zbf}_0 - \sigma\noise\Ebf\right) +  \left(\hh\left[\mpostalp\!\left(\Abf_n \phibf\right)\right]-\phibf\right)\frac{(\Abf\phibf)^\top\left(\Kbis\right)^{-1}\left(\ybf - \Abf_n\tilde{\Zbf}_0 - \sigma\noise\Ebf\right)}{(\Abf_n\phibf)^\top\left(\Kbis\right)^{-1}\Abf_n\phibf}\\
=&\tilde{\Zbf}_0 + \mpost\left(\ybf - \Abf_n\tilde{\Zbf}_0 - \sigma\noise\Ebf\right) - \left(\phibf-\sigbf\Abf_n^\top\left(\Kbis\right)^{-1}\Abf_n\phibf\right)\frac{(\Abf\phibf)^\top\left(\Kbis\right)^{-1}\left(\ybf - \Abf_n\tilde{\Zbf}_0 - \sigma\noise\Ebf\right)}{(\Abf_n\phibf)^\top\left(\Kbis\right)^{-1}\Abf_n\phibf}\\
=&\tilde{\Zbf}_0 + \sigbf\Abf_n^\top\left(\Kbis\right)^{-1}\left(\ybf - \Abf_n\tilde{\Zbf}_0 - \sigma\noise\Ebf\right)
\end{align*}

It comes 
$$\esp(\Zsk) = \msk(\ybf)$$ and
\begin{align*}
\cov(\Zsk) &= \cov\left(\tilde{\Zbf}_0\right)- 2\cov\left(\tilde{\Zbf}_0, \Abf_n\tilde{\Zbf}_0 +\sigma\noise\Ebf\right)\left(\Kbis\right)^{-1}\Abf_n^\top\sigbf \\ &+\sigbf\Abf_n^\top\left(\Kbis\right)^{-1}\cov\left(\Abf_n\tilde{\Zbf}_0+\sigma\noise\Ebf\right)\left(\Kbis\right)^{-1}\Abf_n^\top\sigbf\\
&=\sigma^2\sigbf - 2\sigma^2\sigbf\Abf_n^\top \left(\Kbis\right)^{-1}\Abf_n^\top\sigbf + \sigma^2 \sigbf\Abf_n^\top\left(\Kbis\right)^{-1}\Abf_n^\top\sigbf\\
&= \covsk
\end{align*}
\end{proof}

Using this result, Algorithm~\ref{algo_sk} details the procedure to generate simulations of $\left(\Zbf \mid \Yobs^{(m)} = \ybf, A = \amnoise\right).$

\begin{algorithm}[H]
\caption{Generation of simple kriging simulations}\label{algo_sk}
\begin{algorithmic}[1]
\Require $\Mbf, \Fbf,\alpha, \noise,\sigma,\ybf$, a vector $\gaussbfm$ of $m$ independent standard Gaussian components
\State Generate a simulation $\zbf_0$, centered with covariance matrix $\sigma^2\sigbf$, using Algorithm~\ref{algo_simuprior}
\If{$\noise = 0$} 
\State \textbf{Require :} $I$
\State Compute $\mpostalp(\ybf)$, $\mpostalp(\left[\zbf_0\right]_{I})$ and $\mpostalp\left(\Abf_n\phibf\right)$ using Algorithm~\ref{algo_malp1}
\State Compute $\zsk$ Proposition~\ref{prop_sk} 
\Else 
\State \textbf{Require :} $\Abf_n$ and $\gaussbfn$ a vector of $n$ independent standard Gaussian components
\State Compute $\mpostalp(\ybf)$, $\mpostalp(\Abf_n\zbf_0+\sigma\noise\gaussbfn)$ and $\mpostalp\left(\Abf_n\phibf\right)$ using Algorithm~\ref{algo_malp2}
\State Compute $\zsk$ Proposition~\ref{prop_sk}  
\EndIf
\State \Return $\zsk$
\end{algorithmic}
\end{algorithm}

\section{Spherical harmonics on the Earth}\label{app_harmonics}

Here we show that the results obtained using the kernel $K_1$ on the sphere, based on the well-known spherical harmonics, are really close to what we obtain with our finite-element approximation. Figure~\ref{pred_harmo} provides the SST prediction obtained at the triangulation nodes on the Earth, while Figure~\ref{uncert_harmo} shows the uncertainty associated with the prediction along the longitude closest to $170^\circ$.

\begin{figure}[h]
\centering
\includegraphics[width=\textwidth]{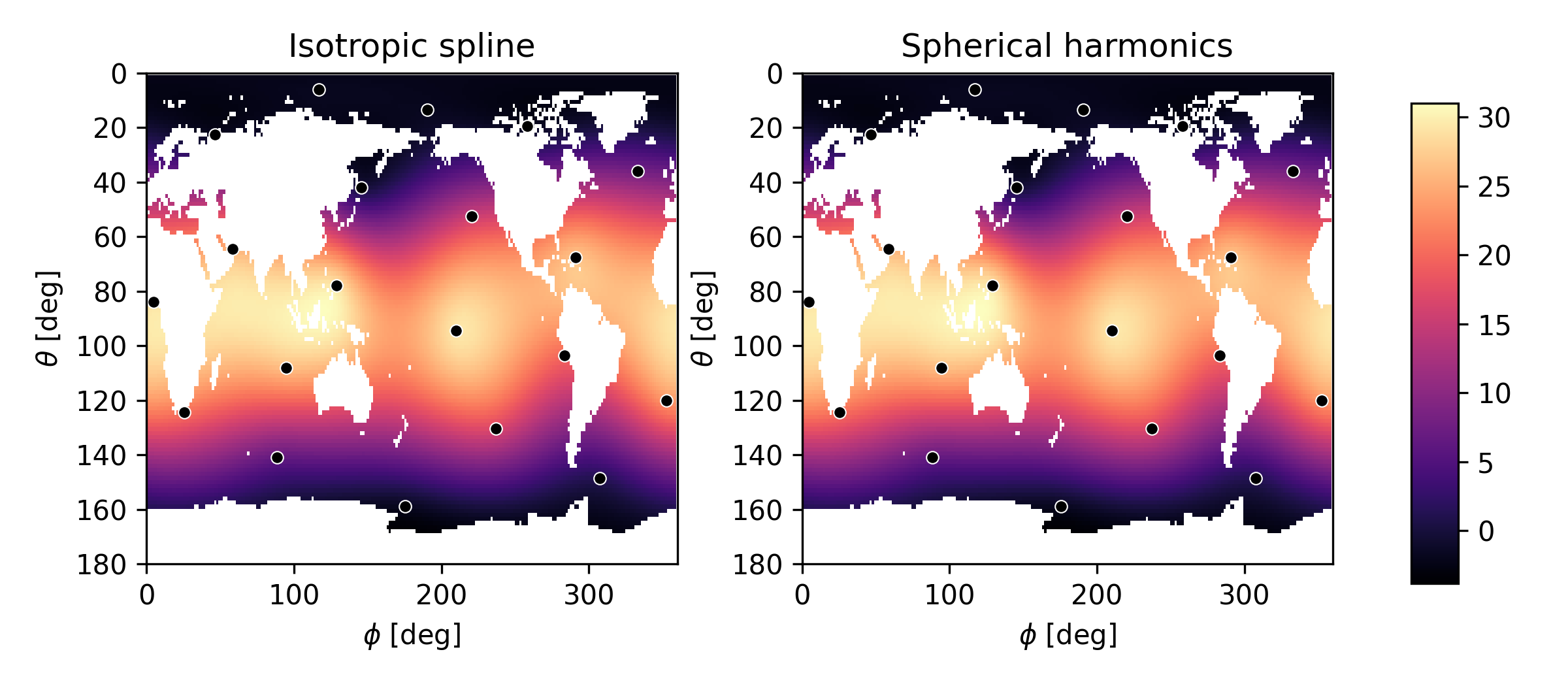}
\caption{Predictions of the SST ($^\circ C$) on Earth, illustrated in 2D using the spherical coordinates $(\theta,\phi).$ $n=20$ observation points are shown as black dots. 
Left: Finite-element approximation of the isotropic spline. Right : Spline using spherical harmonics. }
\label{pred_harmo}
\end{figure} 

\begin{figure}
    \centering
    
    \begin{subfigure}[t]{0.48\textwidth}
        \centering
        \includegraphics[width=\linewidth]{images/simus_iso.png}
\caption{Predictive distribution with the finite-element isotropic spline.}
    \end{subfigure}
    \hfill
    \begin{subfigure}[t]{0.48\textwidth}
        \centering
        \includegraphics[width=\linewidth]{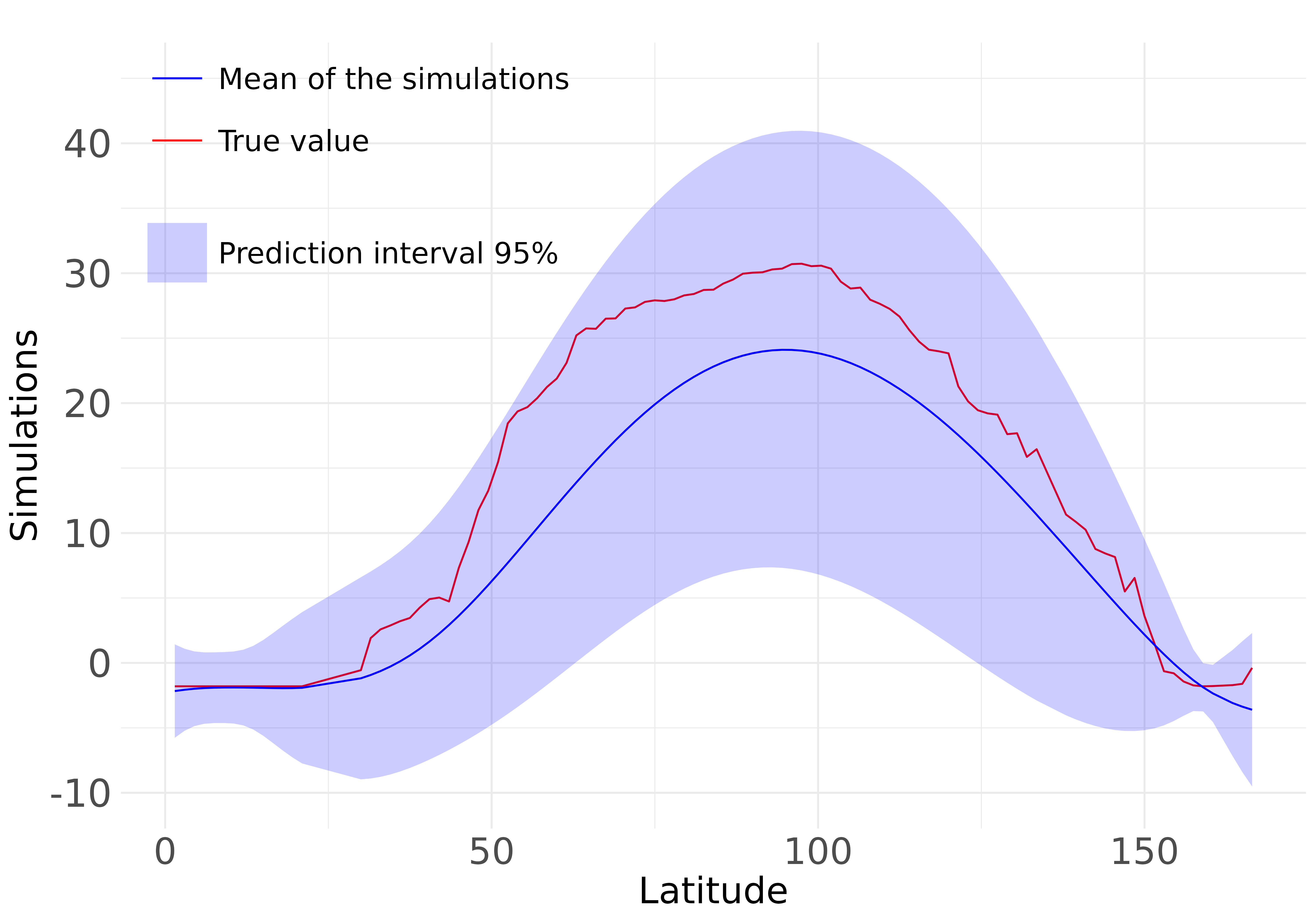}
        \caption{Predictive distribution with the kernel based on the spherical harmonics.}
    \end{subfigure}
    
    \caption{Comparison between the uncertainty quantification in SST prediction on the Earth using the finite-element approximation or the kernel based on the spherical harmonics, along the longitude closest to $170^\circ$.}
    \label{uncert_harmo}
\end{figure}



\end{document}